\documentclass[10pt,journal,compsoc]{IEEEtran}
\usepackage{graphicx,url, amsmath, arrayjobx}
\usepackage{float}
\usepackage[nocompress]{cite}
\usepackage{amsthm, amssymb}
\usepackage{soul}
\usepackage{color}
\usepackage{algorithmic}
\usepackage{caption}
\usepackage{setspace}
\usepackage{subfigure}

\usepackage[linesnumbered, boxed, commentsnumbered, ruled]{algorithm2e}
\usepackage{amsfonts}
\newtheorem{defn}{Definition}[section]

\newtheorem{lemma}{Lemma}[section]
\newtheorem{theorem}{Theorem}[section]

\newcommand{\RNum}[1]{\uppercase\expandafter{\romannumeral #1\relax}}

\allowdisplaybreaks[1]
\ifCLASSOPTIONcompsoc
  \usepackage[nocompress]{cite}
\else
  \usepackage{cite}
\fi

\ifCLASSINFOpdf
\else
\fi
\begin{document}
%\doublespacing
\raggedbottom
\title{Design and Evaluation of A Data Partitioning-Based Intrusion Management Architecture for Database Systems}

%\author{Muhamad~Felemban,~\IEEEmembership{Student Member,~IEEE,} Yahya Javed, Jason Kobes, Thamir Qadah \\
%        Arif~Ghafoor,~\IEEEmembership{Fellow,~IEEE,} and Walid Aref,~\IEEEmembership{Fellow,~IEEE}% <-this % stops a space
\author{Muhamad~Felemban, Yahya Javed, Jason Kobes, Thamir Qadah \\
        Arif~Ghafoor, and Walid Aref% <-this % stops a space
\IEEEcompsocitemizethanks{\IEEEcompsocthanksitem Muhamad Felemban, Yahya Javed, Thamir Qadah, and Arif Ghafoor are with the School of Electrical and Computer Engineering at Purdue University, West Lafayette, IN 47906..\protect\\
}
\IEEEcompsocitemizethanks{\IEEEcompsocthanksitem Walid Aref is with the School of Computer Science at Purdue University, West Lafayette, IN 47906..\protect\\
}
\IEEEcompsocitemizethanks{\IEEEcompsocthanksitem Jason Kobes is with Northrop Grumman \protect\\
}
}

\markboth{}%
{Shell \MakeLowercase{\textit{et al.}}: Partition-Based Intrusion Management for OLTP  Systems }
\IEEEtitleabstractindextext{%
\begin{abstract}
Data-intensive applications exhibit increasing reliance on Database Management Systems (DBMSs, for short). With the growing cyber-security threats to government and commercial infrastructures, the need to develop high resilient cyber systems is becoming increasingly important. Cyber-attacks on DBMSs include intrusion attacks that may result in severe degradation in performance. Several efforts have been directed towards designing an integrated management system to detect, respond, and recover from malicious attacks. In this paper, we propose a data \underline{P}artitioning-based \underline{I}ntrusion \underline{M}anagement \underline{S}ystem (PIMS, for short) that can endure intense malicious intrusion attacks on DBMS. The novelty in PIMS is the ability to contain the damage into data partitions, termed Intrusion Boundaries (IBs, for short). The IB Demarcation Problem (IBDP, for short) is formulated as a mixed integer nonlinear programming. We prove that IBDP is NP-hard. Accordingly, two heuristic solutions for IBDP are introduced. The proposed architecture for PIMS includes novel IB-centric response and recovery mechanisms, which executes compensating transactions. PIMS is prototyped within PostgreSQL, an open-source DBMS. Finally, empirical and experimental performance evaluation of PIMS are conducted to demonstrate that intelligent partitioning of data tuples improves the overall availability of the DBMS under intrusion attacks.
\end{abstract}

}

\maketitle

\IEEEdisplaynontitleabstractindextext
\IEEEpeerreviewmaketitle

\section{Introduction}
%\IEEEraisesectionheading{\section{Introduction}\label{sec:introduction}}

Data-intensive applications exhibit increasing reliance on efficient and scalable Database Management Systems (DBMSs). Examples of these applications are abound in the domain of banking, manufacturing, health care, and enterprise applications \cite{chen2014data}. Since data is the most valuable asset in organizations, it is crucial to design attack-resilient DBMSs to protect the confidentiality, integrity, and availability of the data in the presence of Cyber attacks \cite{stoneburner2002sp,bertino2005database}. Although research in database security has made significant progress in protecting from Cyber attacks, applications and infrastructures are still exposed to a large number of vulnerabilities. Even a single intrusion can cause catastrophic cascading effects due to data dependency and application interoperability. Therefore, a holistic approach for designing an intrusion management mechanism that includes intrusion detection, response, and recovery is needed \cite{kamra2011design,ammann2002recovery}.

An Intrusion Detection System (IDS) is integrated with the DBMS to prevent Cyber attacks. The objective of an IDS is to monitor and detect illegal access and malicious actions that take place in the database. However, an IDS is not designed to repair the damages caused by successful attacks. IDS is often integrated with intrusion response and recovery mechanisms to alleviate the damage caused by the malicious attacks \cite{kamra2011design}. Several efforts have been directed towards developing dynamic damage tracking approach to perform on-the-fly damage repair, for example intrusion-tolerant database systems \cite{liu2004design,ammann2002recovery}. However, such systems have limitations in the ability to maintain high availability under severe intrusion attacks. One of the limitations for these systems is the prolonged recovery time due inter-transaction dependency. 

{\bf{Motivating Example:}} \textit{Fig. \ref{fig:motsc} gives an example scenario for a Banking system with three benign users (B,C, and D), and a malevolent user (A). User A executes a malicious transaction that updates accounts $X$ and $Y$ with incorrect amounts of money as illustrated in Fig. \ref{fig:motsc}. Then, Users B and C withdraw from accounts $X$ and $Y$, respectively. IDS detects the malicious transaction executed by User A and triggers an alert to the database security administrator that temporarily blocks the incoming transactions and starts the recovery procedure. Meanwhile, User D attempts to access account $Z$. However, this request is denied until the damage is recovered. When the recovery transaction is finished, the accounts of Users A and B are compensated, i.e., withdraws money from account $X$ and credits account $Y$.}

In the above example, the recovery time depends on the number of dependent transactions that are executed before the IDS detects the malicious transaction. Consequently, the availability of the DBMS is impaired when the recovery procedure takes a long duration. Therefore, it is important to contain the damage once malicious transactions are detected. Containment of the damage can be achieved by tracking the inter-transaction dependencies and devising a fast confinement strategy to contain the damage.

In this paper, we propose a new real-time response and recovery architecture, termed Partition-based Intrusion Management System (PIMS), for DBMSs. We assume that existing IDS, e.g., \cite{kamra2011design,shebaro2013postgresql}, can be integrated with PIMS. PIMS is based on an adaptive access and admission control mechanism that responds to intrusions by selectively blocking segments of data that have been affected by the intrusion. The access control mechanism provides a fine-grained control policy to allow graceful degradation of database services in order to maintain a desired level of availability while the system is undergoing through impending attacks. The unique feature of PIMS is the deployment of a data partitioning technique to confine the damage and improve the availability of the system. 

The contribution of this paper is summarized as follows. First, we propose the concept of Intrusion Boundary (IB) that defines the extent of the damage over the set of transactional workload. We formulate the IB demarcation as an optimization problem as a Mixed Integer Non-Linear Programming model (MINLP). The output of the optimization problem is a balanced IB assignment of transactions to partitions with minimum overlapping. We prove that the IB demarcation problem is NP-hard. Accordingly, we introduce two heuristics to provide a polynomial time solution. Finally, we introduce response and recovery mechanisms that use the proposed IB assignment to improve the intrusion response and recovery in terms of availability and response time. 

The rest of the paper is organized as follows. Section \ref{sec:bg} presents relevant background. Section \ref{sec:ibf} presents the definition and formulation of the IB demarcation problem, its hardness proof, and present two new heuristics. Section \ref{sec:IMS} describes the design and implementation of PIMS. The performance evaluation and the experimental results of PIMS are presented in Section \ref{sec:exp}. The related work is presented in Section \ref{sec:related}. Finally, Section \ref{sec:conc} concludes the paper.

\begin{figure}[t!]
        \centering
		\includegraphics[width=\columnwidth,scale=0.5]{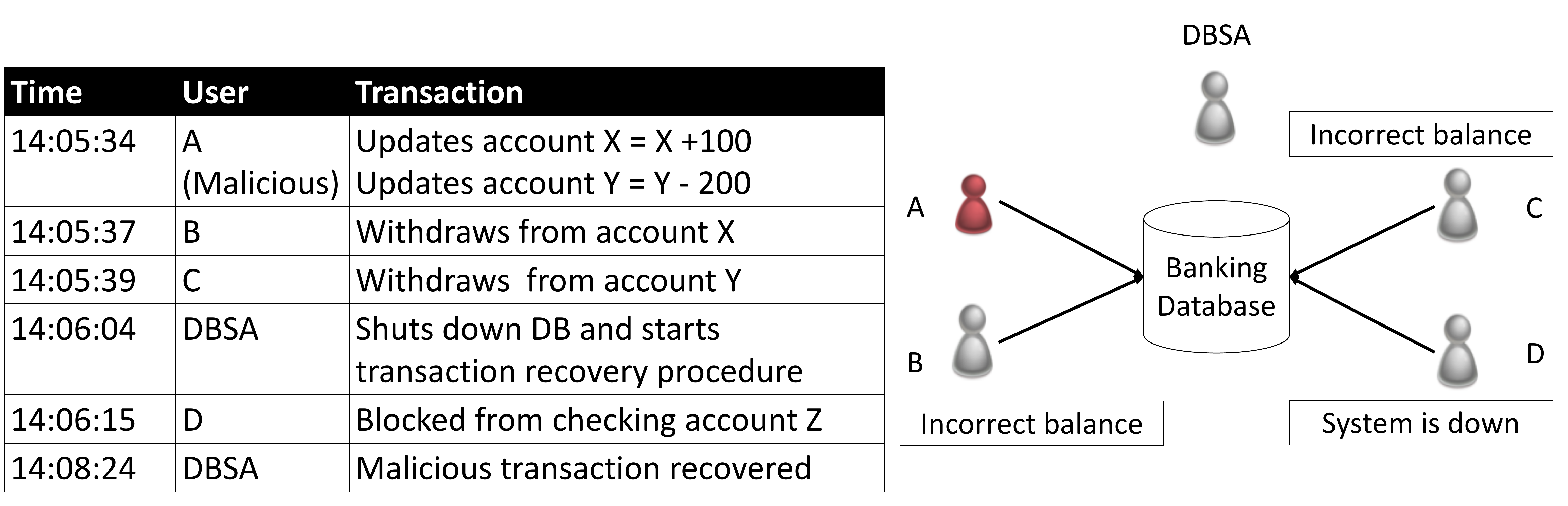}
                \caption{Motivating example.}
                  \label{fig:motsc}
                  \vspace{-3 mm}
\end{figure}

%{\color{red} Example for sequence of transactions where a malicious transaction could affect both availability and integrity\\
%Example where a direct edge between two boundary objects could make collateral damage propgation}
%

\section{Background}
\label{sec:bg}

In this section, we present the database and transaction model and explain the threat model. In addition, we present an overview of the state-of-the-art malicious transaction recovery mechanisms in DBMSs.
\begin{figure}[!t]
  \centering
  \subfigure[Transaction history $H_1$]{\label{fig:h1}   \includegraphics[width=\columnwidth]{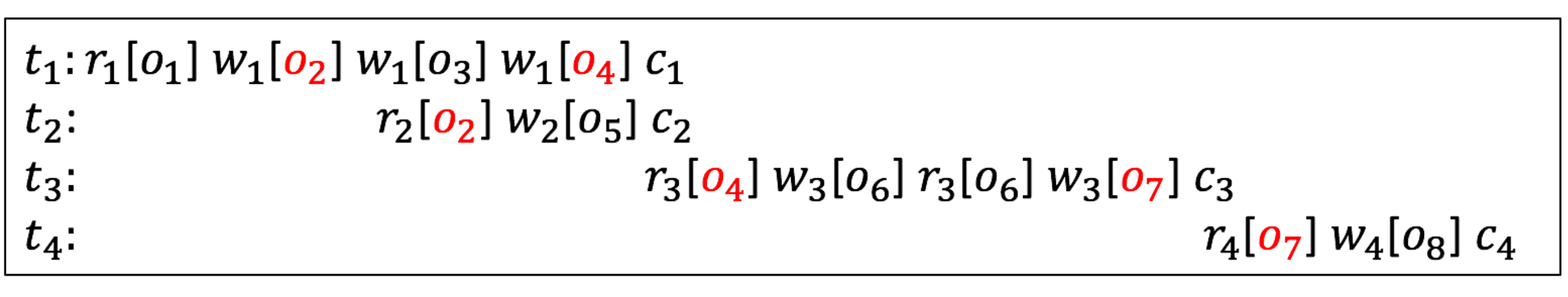}} \\
  \subfigure[Precedence graph]{\label{fig:tdg}\includegraphics[width=0.3\columnwidth]{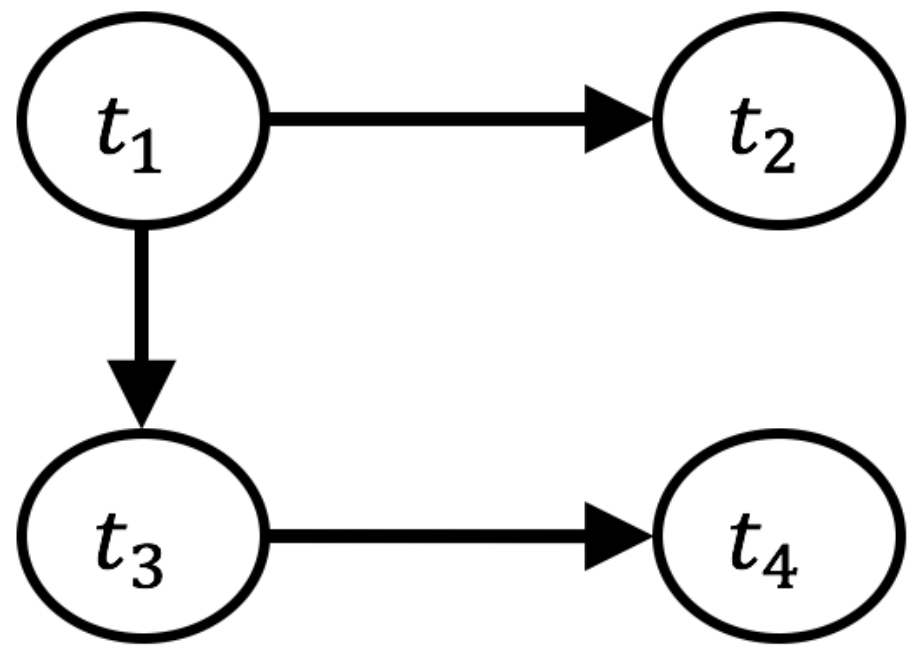}}
  \hskip 0.5truein
  \caption{Transaction history and PG.}
  \label{fig:hist}
                    \vspace{-3 mm}
\end{figure}
\subsection{Database and Transaction Model }
\label{sec:dtmodel}
A database is a set of $n$ tuples denoted by $D = \{o_1,o_2,\ldots,o_n\}$, where $o_i$ refers to a tuple. A transaction $t_i$ is a partially ordered set of operations, with a partial ordering relation $<_i$ \cite{bernstein1987concurrency}. The read and write (update) operations on a tuple, say $o \in D$, of $t_i$ are denoted by $r_{i}[o]$ and $w_{i}[o]$, respectively. The set of tuples that are read and updated by a transaction $t_i$ are denoted by $R_{i}$ and $W_{i}$, respectively. Formally, $t_i$ is defined as follows \cite{bernstein1987concurrency}:
\begin{enumerate}
\item $t_i \subseteq \{ w_i [o],r_i [o] | o \in D\} \cup \{a_i,c_i\}$
\item if $w_i[o], r_i[o] \in t_i,$ then either $w_i[o] <_i r_i[o] $ or $r_i[o] <_i w_i[o]$
\item $a_i \in t_i$ if and only if $c_i \notin t_i$, and
\item for any operation $p \in t_i$, $p <_i t$ where $t$ is either $a_i$ or $c_i$ (whichever in $t_i$)
\end{enumerate}
where $a_i$ and $c_i$ denote the \textit{abort} and \textit{commit} operations of $t_i$, respectively. In essence, Condition (1) defines the types of operations in $t_i$. Condition (2) requires that the order of execution of Read and Write operations on a data item is specified by $<_i$. Condition (3) says that $t_i$ either commits or aborts, while Condition (4) specifies that \textit{commit} (or \textit{abort}) must follow all other operations. 

%For a set of $m$ transactions $T = \{t_1,t_2,...,t_m \}$, a complete history $H$ over $T$ is a partial order with the ordering relation $<_H$, where $H = \cup_{i=1}^m t_i$ and $<_H \supseteq \cup_{i=1}^m <_i$. In other words, the execution history represented by $H$ involves all transactions in $T$ and matches all operation orderings that are specified within each transaction. Two transactions $t_i$ and $t_j$ in $H$ are \textit{dependent} if (i) $t_i$ is directly followed by  $t_j$, i.e., $t_i <_H t_j$, and (ii) $W_{t_i} \cap R_{t_j} \neq \emptyset$. The dependency between $t_i$ and $t_j$ is denoted by $t_i \rightarrow t_j$ and reads "$t_j$ depends on $t_i$". In general,  $t_j$ depends on $t_i$ if (i) $t_i <_H t_{k_1} <_H \ldots <_H t_{k_n} <_H t_j$, and (ii) $(W_{t_i} - \cup_{l=1,..,n} W_{t_{k_l}}) \cap R_{t_j} \neq \emptyset$. This indirect dependency is represented as $t_i \rightarrow^{\propto} t_j$.

For a set of $m$ transactions $T = \{t_1,t_2,...,t_m \}$, a complete history $H$ over $T$ is a partial order with the ordering relation $<_H$, where $H = \cup_{i=1}^m t_i$ and $<_H \supseteq \cup_{i=1}^m <_i$. In other words, the execution history represented by $H$ involves all transactions in $T$ and matches all operation orderings that are specified within each transaction. Two transactions $t_i$ and $t_j$ in $H$ are \textit{dependent} if (i) $t_i$ is directly followed by  $t_j$, i.e., $t_i <_H t_j$, and (ii) $W_{t_i} \cap R_{t_j} \neq \emptyset$. The dependency between $t_i$ and $t_j$ is read  "$t_j$ depends on $t_i$". In general,  $t_j$ depends on $t_i$ if (i) $t_i <_H t_{k_1} <_H \ldots <_H t_{k_n} <_H t_j$, and (ii) $(W_{t_i} - \cup_{l=1,..,n} W_{t_{k_l}}) \cap R_{t_j} \neq \emptyset$.

The dependencies among $T$ in the history $H$ is modeled using the transaction Precedence Graph (PG) \cite{bernstein1987concurrency}. PG is a directed graph $PG=\{ V,E\}$, where $V$ is a set of nodes; each node representing a committed transaction in $H$, and $E$ is a set of edges; each edge representing a dependency between two transactions. In other words, an edge between two transactions $t_i$ and $t_j$ exists if $t_j$ depends on $t_i$. Fig. \ref{fig:h1} illustrates a History $H_1$ over $T = \{ t_1,t_2,t_3,t_4\}$. $H_1$'s corresponding PG is in Fig. \ref{fig:tdg}. $t_2$ depends on $t_1$ because $t_1$ updates $o_2$ that is later read by $t_2$. Similarly, $t_3$ depends on $t_1$ and $t_4$ depends on $t_3$. 

%PG is used to check if $H$ is \textit{conflict serializable} by checking for cycles in the graph. 

\subsection{Threat Model}
In this paper, our focus is on the data corruption caused by transaction-level attacks in the DBMS. These attacks can be manifested either through masquerade access or by exploiting application vulnerabilities, e.g., SQL injection, Cross Site Scripting (XSS), and Cross-Site Request Forgery (CSRF) \cite{srinivasan2017web}. A transaction, say $t_m$, is \textbf{malicious} if it tampers the database by updating one or more tuples with incorrect data. In this context, a malicious transaction corrupts the data due to either an attack or through a user fault. A transaction, say $t_a$, is \textbf{affected} if $t_a$ directly (or indirectly) depends on a malicious or an affected transaction, i.e., $W_{t_m} \cap R_{t_a} \neq \emptyset$ (or $W_{t_j} \cap R_{t_a} \neq \emptyset$ such that $W_{t_m} \cap R_{t_j} \neq \emptyset \forall t_j \in T$). All malicious and affected transactions are invalid transactions. The execution of an invalid transaction takes the DBMS into an invalid state. For example, assume that Transaction $t_1$ in Fig. \ref{fig:tdg} is malicious. Then, both $t_2$ and $t_3$ are affected transactions. We assume that a malicious transaction does not depend on other transactions, i.e., a malicious transaction cannot be an affected transaction. Furthermore, we assume that we can undo the effects of committed transactions.

Throughout the rest of the paper, we deal with integrity attacks in which either an authorized or non-authorized user intentionally tampers with the data by injecting wrong values into some tuples. This attack can temporarily impair the availability of the DBMS during the recovery period. We rely on the existence of an IDS to detect the malicious transactions \cite{shebaro2013postgresql}. Note that the IDS alarm is received after the malicious transactions commit. We assume that the presence of an access control policy is sufficient to prevent any confidentiality attacks. Moreover, we assume that there are security countermeasures to prevent other availability attacks, e.g., Denial-of-Service (DoS) \cite{bertino2005database}.

\subsection{Recovery from Malicious Transactions}

Handling the recovery of malicious transactions requires \textit{undoing} the committed malicious and affected transactions. There are two common approaches to undo committed transactions: \textit{rollback} and \textit{compensation}. The rollback approach is to roll back all the operations performed by the committed transaction to a point that is free of damage \cite{mohan1992efficient}. On the other hand, the compensation approach seeks to selectively undo committed transactions without rolling back the state of the DBMS into a previous state \cite{korth1990formal}.

%On the other hand, the compensation approach seeks to selectively undo committed transactions without rolling back the state of the DBMS into a previous state \cite{korth1990formal}. There are two types of compensation: action-oriented and effect-oriented. The action-oriented approach compensates the action (or set of actions) performed by the committed transaction. On the other hand, the effect-oriented approach compensates the actions performed by the committed transaction and all actions performed by the subsequent transactions. PIMS is aligned with the effect-oriented compensation approach to recover the damage caused by the malicious transactions and the subsequent affected transactions. 

PIMS is relevant to the work in \cite{liu2004design,ammann2002recovery}, in which a real-time approach tracks the inter-transaction dependency, marks the affected transactions, and repairs the damage. The recovery is perfromed in two stages: damage assessment and damage repair. In the damage assessment stage, the complete and correct set of affected transactions is identified. This stage is challenging because the assessment is conducted on-the-fly while the system is processed other concurrent transactions. Due to transaction dependency, the damage might spread. Therefore, the damage assessment is terminated whenever there are no more transactions that cause spread of the damage. Once the damage assessment stage is completed, all the identified affected transaction are repaired by rolling them back without affecting the other transactions in the system. PIMS's recovery approach is different than \cite{liu2004design,ammann2002recovery} in a sense that PIMS aggressively terminates the damage spread by stopping the execution of the concurrent transactions momentarily during the damage assessment process. Moreover, PIMS adopts data partitioning scheme that provides a proactive damage confinement mechanism by designating a group of tuples to prevent the spread of the damage into the entire database. 

\section{Intrusion Boundary Demarcation using Data Partitioning}
\label{sec:ibf}
In this section, we present a data-level model to represent intra-transaction and inter-transaction dependencies. Accordingly, we define the IB demarcation problem, and formulate it as MINLP optimization problem. Finally, we prove the problem's hardness and present two efficient heuristics to solve the IB demarcation problem. 

\subsection{The IB Demarcation Problem}
The objective of the IB demarcation problem is to partition the tuples into $k$ partitions, i.e., IBs, with minimum overlap. The advantage of the IB demarcation is to confine the damage into a single IB, and thus increases the data availability in the presence of attacks. Notice that, the overlap among partitions depends on the inter-transaction dependency. For workloads with independent transactions, the demarcation results in non-overlapping IBs. On the other hand, for workloads with high inter-transaction dependencies the IB demarcation results in overlapping IBs. The IB demarcation is defined as follows. 

\begin{defn}
(Intrusion Boundary Demarcation Problem (IBDP)) Given the set of transactions $T$ over a set of $n$ tuples, IBDP is to assign the transactions into $k$ IBs such that the overlap among the IBs is minimized and the sizes of the IBs are almost equal.
\end{defn}

%\begin{defn}
%(Balanced Min-Overlapping Intrusion Boundary Demarcation Problem (BMOIBDP, for short)) Given the DDG created using a set of transactions $T$ over a set of $n$ tuples, IBDP is to partition DDG into $k$ IBs such that the overlap among the IBs is minimized and the sizes of the IBs are almost equal.
%\end{defn}

%\begin{defn}
%(General Intrusion Boundary Demarcation Problem (GIBDP, for short)) Given the DDG created using a set of transactions $T$ over a set of $n$ tuples, IBDP is to partition DDG into $k$ IBs such that the set of objective functions ($f_1, f_2, \ldots, f_q$) are minimized.
%\end{defn}

\subsubsection{Problem Formulation}
The demarcation of IBs is controlled by the dependencies among the tuples that are accessed by multiple transactions. We formulate IBDP as a dual-objective MINLP. The first objective function focuses on damage containment that minimizes the overlap among IBs. In order to define this objective function, we define the set of Boundary Tuples (BT) as follows. 
\begin{defn} (Boundary tuple) A tuple, say $o \in D$, termed a \textbf{boundary} tuples if $o$ is assigned to two or more IBs. The set of boundary tuples in the IB assignment is denoted by $BT$.  Observe that BT $\subseteq ST $. 
\end{defn}

%For example, the set of boundary tuples in Fig. \ref{fig:ddg} is $BT=\{o_4\}$.
It is sufficient to minimize the number of boundary tuples in order to minimize the overlap among IBs. Thus, the objective function can be defined as follows. 
\begin{equation}
f_1(B) = \sum_{i=1}^n b_i
\label{eq:f1}
\end{equation}
where B is the boundary tuples vector that indicates if a tuple is boundary, i.e., $b_k$=1 if the $i^{th}$ tuple is a boundary tuple, and 0, otherwise (notice that $|B|$ = n). The intuition behind minimizing the number of boundary tuples is to limit the damage propagation across IBs. However, Equation \ref{eq:f1} is oblivious to the number of IBs that share a boundary tuple. This number is termed the \textit{degree of sharing}. A boundary tuple that is shared between 2 IBs has less risk of damage propagation as compared to a boundary tuple that is shared among multiple IBs. The degree of sharing is incorporated in Equation \ref{eq:f1} in the following way.  
\begin{equation}
f_1(B,\mathit{TuIB}) = \sum_{i=1}^n b_i \left( \sum_{j=1}^k \mathit{TuIB}_{ij} - 1\right) 
\label{eq:f1_2}
\end{equation}
where $\mathit{TuIB}$ is the IB assignment matrix for the tuples, i.e., $\mathit{TuIB}_{ij}$ = 1 if the $i^{th}$ tuple is assigned to the $j^{th}$ IB, and 0, otherwise. The sum $\sum_{j=1}^k \mathit{TuIB}_{ij}$ is the number of IBs to which the $i^{th}$ tuple is assigned, whereas the sum $\sum_{i=1}^n \mathit{TuIB}_{ij}$ is the size of the $j^{th}$ IB. 

The second objective function focuses on the overall availability. The goal of this objective function is to prevent skewed IB assignment by balancing the sizes of the IBs. Formally, the objective of balancing the sizes of IBs is defined as follows.
\begin{equation}
f_2(\mathit{TuIB}) =  \sqrt{\sum_{i=1}^k \sum_{j>i}^k \left( \sum_{\ell=1}^n \mathit{TuIB}_{\ell i} - \sum_{\ell=1}^n \mathit{TuIB}_{\ell j} \right)^2}
\label{eq:f2_1}
\end{equation}
Let $\mathit{TrIB}$ be the IB assignment matrix for the transactions, i.e., $\mathit{TrIB}_{ij}=1$ if the $i^{th}$ transaction is assigned to the $j^{th}$ IB, and 0, otherwise. Similarly, Let $\mathit{TuTr}$ be a binary matrix representing the association of tuples to the transactions, i.e., $\mathit{TuTr}_{ij}=1$ if the $i^{th}$ tuples is accessed by the $j^{th}$ transaction, and 0, otherwise. Although a transaction might span multiple IBs, each transaction must be fully contained within a single IB. The intuition is that by containing a transaction within an IB, the damage is confined in that IB if the transaction is detected as malicious. Accordingly, IBDP is formulated as MINLP using objective functions $f_1$ and $f_2$ as follows. 

%\begin{align}
%& \underset{\mathit{TUA},\mathit{TRA},B}{\text{Minimize }} \nonumber
%& & f_1(\mathit{TUA},B) + f_2(\mathit{TUA})\\ \nonumber
%& \text{subject to} \\ 
%& & & \sum_{j=1}^k \mathit{TUA}_{ij} -1 \leq k b_i , \ \ \ \forall i \in \{ 1,\ldots,n\} \label{eq:c1} \\ 
%& & & 2- \sum_{j=1}^k \mathit{TUA}_{ij} \leq k (1-b_i) ,\nonumber \\ 
%& & &  \ \ \ \forall i \in \{ 1,\ldots,n\} \label{eq:c2} \\ 
%& & & \sum_{i=1}^{n} \mathit{TUTRA}_{i \ell} - \sum_{i=1}^{n} \mathit{TUTRA}_{i \ell} * \mathit{TUA}_{ij} \nonumber \\ 
%& & & \ \ \ \ \ \ \ \ \ \geq n * (1-\mathit{TRA}_{\ell j}) \label{eq:c3} \nonumber \\ 
%& & & \ \ \ \ \ \forall \ell \in \{ 1,\ldots,m\}, j \in \{ 1,\ldots,k\}  \\ 
%& & & \sum_{j=1}^k \mathit{TRA}_{ij} = 1 \ \ \ \ \forall i \in \{ 1,\ldots,m\} \ \ \  \label{eq:c4} \\ 
%& & & C \geq \sum_{i=1}^n \mathit{TUA}_{ij} \geq 1 \ \ \ \ \forall j \in \{ 1,\ldots,k\} \ \ \  \label{eq:c5} \\  
%& & & \mathit{TUA}_{ij} \in \{0,1\}, \mathit{TRA}_{\ell j} \in \{ 0,1\}, b_i \in \{0,1\} \label{eq:c6}  \nonumber \\
%& & & \forall i \in \{ 1,\ldots, n\} \ , \ell \in \{ 1, \ldots,m\},\ \  \\ \nonumber
%& & &  \ j \in \{ 1,\ldots, k\}
%\label{eq:op2}
%\end{align}
                  \vspace{-3 mm}
\begin{align}
& \underset{\mathit{TuIB},\mathit{TrIB},B}{\text{Minimize }} \nonumber
& &  f_1(B,\mathit{TuIB})  +f_2(\mathit{TuIB})\\ \nonumber
& \text{subject to} \\ 
& & & \sum_{j=1}^k \mathit{TuIB}_{ij} -1 \leq k b_i , \ \ \ \forall i \in \{ 1,\ldots,n\} \label{eq:c1} \\ 
& & & 2- \sum_{j=1}^k \mathit{TuIB}_{ij} \leq k (1-b_i) , \nonumber \\
& & & \ \ \ \forall i \in \{ 1,\ldots,n\} \label{eq:c2} \\ 
& & & \sum_{i=1}^{n} \mathit{TuTr}_{i\ell} - \sum_{i=1}^{n} \mathit{TuTr}_{i\ell} * \mathit{TuIB}_{ij} \geq \nonumber \\
& & & 1-\mathit{TrIB}_{\ell j} \label{eq:c3} \\ 
& & & \ \ \ \ \ \forall \ell \in \{ 1,\ldots,m\}, j \in \{ 1,\ldots,k\} \nonumber \\ 
& & & \sum_{j=1}^k \mathit{TrIB}_{ij} = 1 \ \ \ \ \forall i \in \{ 1,\ldots,m\} \ \ \  \label{eq:c4} \\ 
& & & \sum_{i=1}^n \mathit{TuIB}_{ij} \geq 1 \ \ \ \ \forall j \in \{ 1,\ldots,k\} \ \ \  \label{eq:c5} \\  
%& & & \sum_{j=1}^K x_{ij} \geq 1 \ \ \ \forall i \in D \\ 
& & & TuIB_{ij} \in \{0,1\}, TrIB_{ij} \in \{0,1\}, \nonumber \\ 
& & & b_i \in \{0,1\} \label{eq:c6}  \\  \nonumber 
\label{eq:op2}
                  \vspace{-3 mm}
\end{align}

The outputs of the optimization problem are the IB assignment matrix of the tuples $\mathit{TuIB}$, the IBs assignment matrix of the transactions $\mathit{TrIB}$, and the set of boundary tuples $B$. Constraints (\ref{eq:c1}) and (\ref{eq:c2}) collectively check if a tuple is assigned to multiple IBs. Accordingly, the constraints set $b_i =1$ if the $i^{th}$ tuple is boundary, and 0, otherwise. Full containment of a transaction within a single IB is checked by constraints (\ref{eq:c3}) and (\ref{eq:c4}). In particular, a transaction $t_{\ell}$ is assigned to the $j^{th}$ IB only if all the tuples access by $t\_{\ell}$ are assigned to the $j^{th}$ IB. Constraint (\ref{eq:c5}) forces the size of each IB to be at least one tuples, while constraint (\ref{eq:c6}) forces $\mathit{TuIB}$, $\mathit{TrIB}$, and $B$ to be binary matrices.

\begin{theorem}
IBDP is NP-hard.
\label{th:ibdp}
\end{theorem}
\label{app:a}

\subsection{Heuristics for IB Demarcation}
We introduce two efficient greedy-based heuristics to solve IBDP. All algorithms take the set of transaction and the number of IBs as input. The output is a transaction-to-IB assignment. The algorithms start with an empty IB assignment and iteratively assign transactions to IBs based on greedy decisions that optimize the objective functions. The first heuristic is Best-Fit Assignment (BFA) that reduces the number of boundary tuples produced by the IB assignment. The intuition of the assignment is that BFA assigns the transaction to the IB that shares the largest number of shared tuples. The second heuristic is Balanced Assignment (BA) that assigns transactions such that the sizes of all IBs are almost equal. This is achieved by assigning, at each iteration, the transaction to the IB that is the smallest in size. Detailed discussion about each algorithm is presented in the following sections.

\subsection{Best-Fit Assignment (BFA)}

BFA is listed in Algorithm \ref{alg:bfa}. The algorithm starts with the empty IB assignment set $\mathcal{S}$ and empty assigned transactions set $A$ (Line 1). Then, the transactions are sorted based on the number of internal tuples in descending order (Line 2). The sorted transaction set is stored in Set $T$. Then, the algorithm initially assigns the first $k$ transactions to the empty IBs (Lines 3-9). As a result, the overlap among the IBs is minimized. This is correct because the first $k$ transactions in $T$ have the least number of shared tuples. For the remaining transactions in $T$, each transaction $t$ is assigned to the IB that shares the largest number of shared tuples with $t$ (Lines 10-19). In particular, the set of assigned transactions that overlap with $t$ is stored in $NT$ (Line 11). Then, the IBs of each transaction in $NT$ is stored in $NIB$ (Line 12). $NIB$ contains the set of all IBs that share tuples with $t$. If the set $NIB$ is not empty, $t$ is assigned to the IB that has the largest number of shared tuples with $t$ (Line 14-16). Otherwise, $t$ is assigned to the smallest IB in $\mathcal{S}$. 

\begin{lemma}
The complexity of BFA is $O(nm^2 + knm)$.
\label{lem:1}
\end{lemma}

\begin{proof}
\label{app:pf1}
BFA sorts the set of transactions based on the number of internal tuples in each transaction (Line 3). Sorting Set $T$ has a runtime complexity of $\mathcal{O}(m \log m)$. The loop in Lines 3-9 assigns a single transaction to each empty IB. Thus, the loop has a runtime complexity of $\mathcal{O}(k)$. Then, BFA assigns transactions to the best-fit IB (Lines 10-19). The runtime complexity of finding the set $NT$ and $NIB$ is $\mathcal{O}(nm)$ (Lines 11 and 12). Finding $ib_{max}$, the IB that has the largest number of tuples shared with $t$, and assigning $t$ in Lines 13-18 has a runtime complexity of $\mathcal{O}(nk)$. The runtime complexity for assigning all the transactions is $\mathcal{O}(m ( nm + nk))$. The overall complexity of BFA is $\mathcal{O}(m \log m + k + nm^2 + knm)$, i.e., approximately $~ \mathcal{O}(nm^2 + knm)$.
\end{proof}

\begin{algorithm}[t!]
\small
\SetNlSty{normal}{}{.}
\KwIn{$k$}
\KwOut{$\mathcal{S}=\{ib_1,...,ib_k\}$}
$\mathcal{S}=\emptyset$, $A = \emptyset$\\
$T \leftarrow$Sort transactions based on the number of internal tuples in descending order \\ 
\For{$i=1..k$}{
$t \leftarrow$ largest transaction in $T$  \\
$ib_i = \{ t \}$\\ 
$\mathcal{S} = \mathcal{S} \cup \{ ib_i\}$\\
Add $t$ to $A$\\
Remove $t$ from $T$\\
}
\For{$t \in T$}{
Find $NT$ the set of assigned transaction that overlap with $t$\\
Find $NIB$ the set  of IBs that overlap with $t$\\
\uIf{$NIB \neq \emptyset$}{
Find $ib_{max} \in NIB$ that has the largest number of tuples shared with $t$\\
$ib_{max} = ib_{max} \cup t$ \\
}
\Else{
Assign $t$ to the smallest $ib$\\
}
}
\Return{$\mathcal{S}$}
\caption{Best Fit Assignment}	
\label{alg:bfa}
\end{algorithm}

\section{The Architecture of PIMS}
\label{sec:IMS}

In this section, we introduce the architecture of PIMS. The proposed architecture for PIMS is given in Fig. \ref{fig:ims_arch}. PIMS is composed of five components: the IBDP solver, the transactions log, the Admission Controller (AC), the Response Subsystem (RES), and the Recovery Subsystem (REC). In addition, PIMS maintains a Corrupted Tuples Table (CTT) to track the status of the damage caused by the malicious transactions. The IBDP solver generates the IB assignment as discussed in the previous section. The transactions log stores information about the read/write operations of transactions. The log is essential for the recovery procedure to construct compensating transactions. The functionality of AC includes parsing the transactions, regulating the admission of transactions to DBMS, and maintaining the IB assignment table. AC determines whether to block or allow incoming transactions depending on the status of the damage and the assigned IB. When a committed transaction is identified as malicious, RES extracts the commit and detection times of the malicious transaction, approximates the set of corrupted tuples immediately, and stores it in CTT. Although the approximated set can include uncorrupted tuples, the objective is to reduce the risk of executing benign transactions that might propagate the damage unwillingly. Then, REC analyzes the inter-transaction dependency in transactions log and identifies the \textit{correct} and \textit{complete} set of affected transactions. The correct set of affected transactions means that no transactions are falsely identified as affected, while the complete set means that the set contains every affected transactions caused by the attack. We refer to this set as the set of \textbf{Affected Transactions} (AT). When AT is identified, the benign tuples are removed from CTT, and consequently can be accessed by the requesting transactions immediately. On the other hand, the corrupted tuples are blocked and will be recovered by the compensating transactions. The recovered tuples are removed from CTT gradually in order to increase the system availability. In the following sections, we present detailed information about the functionality of each component. 

\label{sec:ims}

\begin{figure}[t!]
        \centering
		\includegraphics[width=\columnwidth,height=5cm]{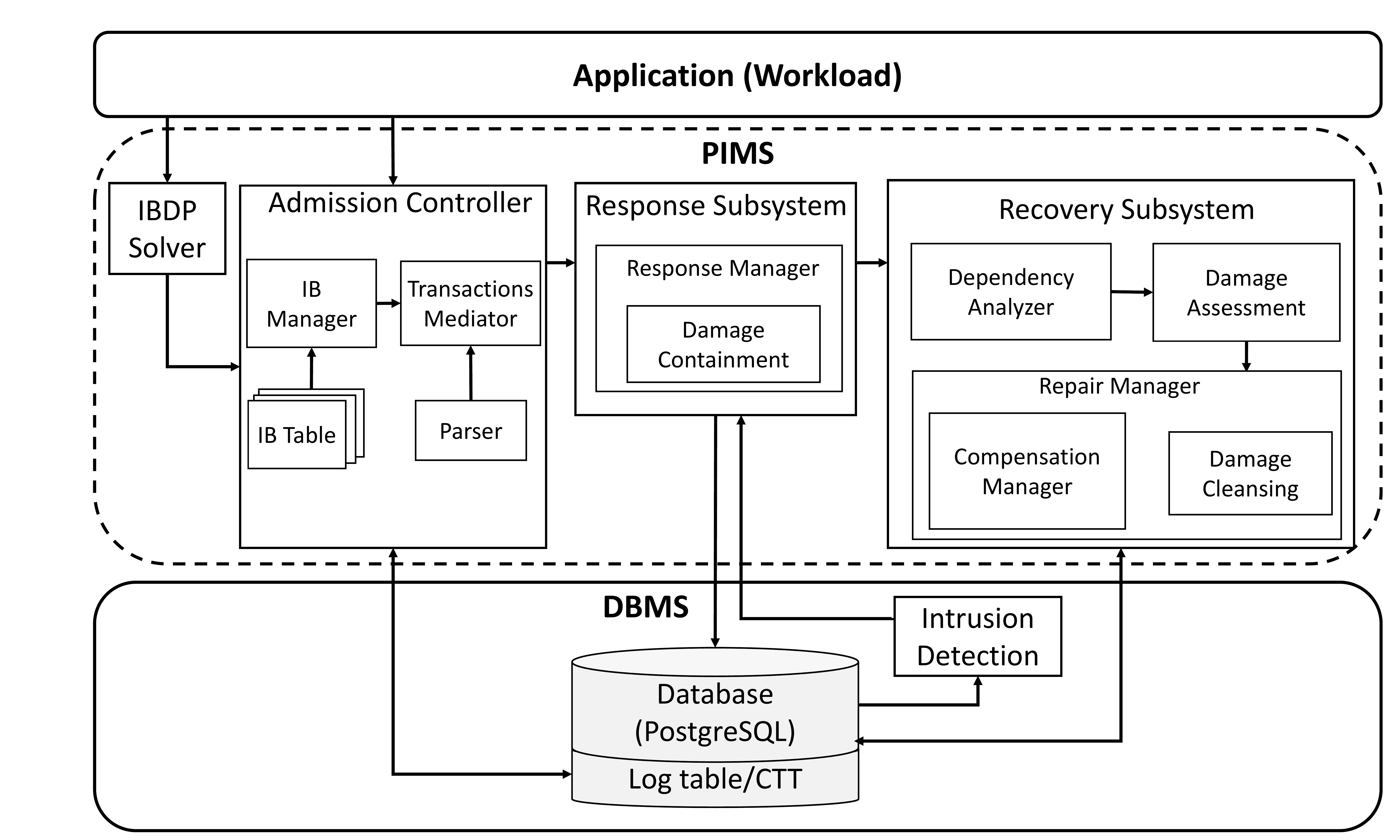}
                \caption{PIMS Architecture.}
                  \label{fig:ims_arch}
                                    \vspace{-3 mm}
\end{figure}

\subsection{Transactions Log}

In order to obtain accurate information about the extent of the damage caused by the malicious transaction, the read/write operations of all transactions need to be logged. However, conventional undo/redo logs and DBMS triggers only record write operations. To address this issue, we have implemented a read/write log that records the transaction ID, the tuple ID, the before-image (the previous value), the after-image (the current value), and the time-stamp for each read and update operation. The transactions log is implemented in the DBMS kernel to reduce the overhead of logging. The transactions log is maintained as a table in the DBMS for efficient retrieval by PIMS. The transactions log is used by AC to check if a transaction reads from corrupted tuples, and by RES to identify AT. 

\subsection{The Admission Controller (AC)}

AC has three subcomponents: the Parser, the IB Manager (IBM), and the Transaction Mediator (TM). The parser extracts the read/write set from transactions. IBM maintains in the IB table information about the IBs including the set of boundary tuples and the transaction-to-IB assignment. The functionality of IBM is to manage the access to the boundary tuples. In essence, the updated boundary tuples are locked by IBM until IDS reports the status of the updating transaction. If the transaction is identified as malicious, then the locked boundary tuples are added to CTT. Otherwise, the locked boundary tuples are released. One problem is that IDS only triggers an alarm when a malicious transaction is detected. We assume that IDS has a detection delay of $\Delta$ $ms$. The boundary tuples are locked for a sufficient time after which the updating transaction can never be detected as malicious. We set PIMS to wait for 1.5 $\times \Delta$ ms before releasing the boundary tuples. The objective of locking the boundary tuples after update is to assure that the damage does not propagate to other IBs. We refer to the process of locking the boundary tuples as \textit{delayed access} mechanism. 

TM checks if the transaction's read/write set contains any tuple that exists in CTT or BT. If a tuple in the read set exists in CTT, the transaction is suspended. On the other hand, if a tuple in the write set exists in CTT, then the transaction is executed, and the tuples are removed from CTT and are excluded from subsequent undo and redo operations. AC is signaled on two events: 1) the corrupted tuples are recovered, or 2) the recovery procedure is completed. When signaled, TM executes the suspended transaction if its read tuples are all recovered. Otherwise, the transaction remains suspended until a new signal is received from RES or REC. If the transaction's read/write set contains a tuple that exists in BT, then the transaction is suspended until the tuples are released by IBM. 

Before executing the transaction, AC acquires the locks associated with the assigned IB(s). The lock is to ensure that no concurrent recovery procedure is running on the IB(s). The lock is released by AC when the transaction is executed successfully. The overall procedure for admission control is listed in Algorithm \ref{alg:adm}. Notice that multiple instances of the admission controller can be executed using multiple threads to increase transaction concurrency. We rely on the available Concurrency APIs to queue the incoming transactions when TM threads are busy waiting for signals. 

\begin{algorithm}[t!]
\small
\SetNlSty{normal}{}{.}
\KwIn{$t_i$}
\KwOut{}
$RW_{t_i} \leftarrow$ the read/write tuples set $t_i$\\
\While{$RW_{t_i} \cap BT \neq \emptyset$}{
Wait until the requested $BT$ is released\\
}
\While{$ RW_{t_i} \cap CTT \neq \emptyset$ }{
Wait until request tuples in CTT are recovered and released\\
}
\tcc{Wait until lock is acquired}
Acquire lock on $IB_{t_i}$ \\
Execute $t_i$\\
Release lock on $IB_{t_i}$\\
\caption{Admission control}	
\label{alg:adm}
\end{algorithm}

\subsection{The Response Subsystem (RES)}

RES is activated when a transaction $t_m$ is detected as malicious by IDS. The objective of RES is to prevent subsequent benign transactions from reading corrupted tuples that are updated by $t_m$, and thus control the spread of the damage. RES collects the time information about $t_m$ from IDS and the transactions log, i.e., the commit timestamp $t_m^c$ and the detection timestamp $t_m^d$. Moreover, RES extracts $IB_{t_{m}}$, the set of IBs that are spanned by $t_m$, from the IB table. Consequently, RES adds all tuples in $IB_{t_m}$ that have been updated between $t_m^c$ and $t_m^d$ to CTT. Notice that the tuples that are updated between $t_m^c$ and $t_m^d$ but not assigned to $IB_{t_m}$ are not added to CTT. Thus, PIMS provides more accurate damage confinement as compared to ITDB, which blocks all tuples on temporal basis only. Nevertheless, CTT might contain uncorrupted tuples updated by benign transactions that are independent from $t_m$. Once enough information about the inter-transaction dependencies between $t_m$ and subsequent transactions is gathered, PIMS removes the uncorrupted tuples from CTT as explained in the next section. The procedure of RES is summarized in Algorithm \ref{alg:res}.
\begin{algorithm}[t!]
\small
\SetNlSty{normal}{}{.}
\KwIn{$t_m, t_m^d$}
\KwOut{Updated $CTT$}
$t_m^c$ $\leftarrow$ get commit time of $t_m$\\
Find $IB_{t_m}$ the set of IBs spanned by $t_m$\\
\For{tuples updated between $t_m^c$ and $t_m^d$}{
\If{tuple is assigned to an $IB \in IB_{t_m}$}{
Add tuple to $CTT$
}
}
\caption{Intrusion Response}	
\label{alg:res}
\end{algorithm}

%After the set of suspicious tuples is identified, IR makes this set available to the mediator at each affected IB. The mediator blocks any incoming transaction that depends on the suspicious set until the suspicious set is recovered. One challenge is to identify the correct suspicious set as fast as possible otherwise incoming transactions can spread the damage while IR is working on identifying the correct set. However, approximating the set in a short time might results in either overestimating or underestimating the suspicious set. In case of overestimating, more tuples are included and hence more transactions are blocked. On other hand, underestimating results in more damage propagation. This phenomena has been studied in \cite{liu2004design} and \cite {ammann2002recovery}. We adopt a similar strategy, in which the suspicious set is approximated by including all updated tuples during $qw$ that are within the affected set of IBs without checking if they actually depends on the malicious transaction. Clearly, this is overestimating of the suspicious set. The advantage is that the identification process is done in short time, thus limiting the damage propagation. Good tuples are then released in later stages once the set actually affected tuples are carefully identified. Details about releasing good tuples will be provided in the next section. 

\subsection{The Recovery Subsystem (REC)}
\label{sec:RES}
%Per IB mechanism 
REC is the core component of PIMS that identifies AT and executes compensating transactions for the corrupted tuples. The compensating transactions perform two operations: \textit{undo} and \textit{redo}. The undo operation unwinds the effect of the malicious transaction and each transaction in AT. By performing the undo operation, the state of the DBMS returns to the state just before the malicious transaction is executed. However, the update operations performed by the subsequent benign transactions are lost. The redo operation preserves the lost updates by re-executing each transaction in AT. Notice that the malicious transaction is not re-executed because its updates are undesirable. 
\begin{algorithm}[t!]
\small
\SetNlSty{normal}{}{.}
\KwIn{$t_m, IB_{t_m}$}
\KwOut{Updated CTT}
Acquire recovery lock in every $IB \in IB_{t_m}$\\ 
Block new transactions in $IB_{t_m}$\\
Wait for currently running transaction in $IB_{t_m}$ to commit\\
%Acquire lock for each $IB_{t_m}$\\
Find the set of affected transaction $AT_{t_m}$ \\ 
Find the $S$ set of all updated tuples by $t_m$ and $t_a \in AT$\\
%Release lock each $IB_{t_m}$\\
Resume new transaction in $IB_{t_m}$\\
\For{$o \in CTT$}{
\If{$o \notin S$}{
Remove $o$ from CTT and flag as valid
}
}
Resume transactions to $IB_{t_m}$\\
\tcc{ Phase \RNum{1}}
Undo $t_m$\\
\For{$T_a \in AT_{t_m}$}{
Undo $T_a$ \\
}
\tcc{ Phase \RNum{2} }
\For{$T_a \in AT_{t_m}$}{
Redo $T_a$ \\
}
\caption{Intrusion Recovery}	
\label{alg:rec}
\end{algorithm}

The overall algorithm for REC is listed in Algorithm \ref{alg:rec}. REC is activated once the response procedure is finished. Thus, the response and recovery transactions are executed serially. Serializing the response and recovery procedures is essential since the recovery mechanism uses CTT that is updated by the response subsystem. REC temporarily blocks new transactions to prevent new transactions from reading the corrupted tuples while REC identifies AT. This is achieved by acquiring locks on $IB_{t_{m}}$. Incoming transactions are blocked by AC. Notice that the active transactions are not preempted. When all active transactions are committed, the \textit{Dependency Analyzer} scans the transactions log starting from $t_m^c$ through the current timestamp in order to find AT. The write set of each transaction in AT is extracted and is added to CTT. The \textit{Damage Assessment} subsystem removes the uncorrupted tuples that have been initially added by RES from CTT.  

At this point, CTT contains the correct and complete set of the corrupted tuples caused by $t_m$. The \textit{Compensation Manager} (CM) executes a sequence of compensating transactions that gradually repair the damage. Damage repair is performed in two phases. In Phase 1, CM executes compensating transactions that unwind the effect of the malicious and affected transactions. CM uses the transactions log table to find the correct version of the corrupted tuples. In particular, CM updates the corrupted tuples with the values of the most recent versions before the execution of the malicious transaction. The compensating transactions are executed in the order at which the malicious and affected transactions are committed. At the end of this phase, \textit{Damage Cleansing} (DC) removes the recovered tuples from CTT and signals AC to resume any blocked transactions. 

In the second phase, CM executes compensating transactions that re-execute each transaction in AT in the same order at which they are committed. The information required to re-execute the transactions, e.g., old balance, is maintained in the transactions log. Once an affected transaction is re-executed successfully, DC removes the set of recovered tuples from CTT and signals AC.  

\subsection{Managing Multiple Malicious Transactions}

The advantage of partitioning the tuples into IBs is to execute concurrent response and recovery procedures on different IBs. The reason is that all transactions in AT are guaranteed to be contained in $IB_{t_m}$. As a result, concurrent recovery procedures do not perform conflicting operations while recovering the corrupted tuples. Multiple instances of response and recovery procedures are executed using multiple threads.
In the case when multiple malicious transactions are detected within the same IB, the recovery procedures need to be coordinated to avoid out-of-order execution of the recovery operations. The strategy is to execute the recovery procedures for the malicious transactions in the same order in which they are detected.  

\section{Performance Evaluation}
\label{sec:exp}
In this section, performance evaluation of the proposed PIMS framework is presented. First, we discuss the synthetic workload used for the performance evaluation. Then, we evaluate the performance of the proposed IB demarcation heuristics. Finally, we conduct extensive system evaluation of PIMS using synthesized workloads, and present the performance results. 

\subsection{Synthetic Transactional Workload}
Several benchmarks have been developed to evaluate the performance of OLTP systems, e.g., TPC-C \cite{tpcc}, \textit{SmallBank} \cite{alomari2008cost}, and YCSB \cite{cooper2010benchmarking}. However, no benchmark has been developed to evaluate the performance of  intrusion management systems in OLTP. In order to evaluate the performance of PIMS, we have developed a malicious transaction workload benchmark that generates long chains of dependent transactions. These long chains of dependent transactions amplify the potential damage that can be caused by the malicious transactions. In essence, executing dependent transactions shortly after malicious transactions induce the spread of damage. Thus, the proposed benchmark allows the assessment of the capability of the intrusion management systems to confine and recover the damage without degrading the performance of the DBMS. 

The proposed benchmark simulates a banking money-transfer application. In essence, the benchmark consists of a single data table, \textit{Checking}, that has two attributes: \textit{id} and \textit{balance}. The benchmark has three types of money transfer transactions: \textit{distribute}, \textit{collect}, and \textit{many-to-many} transfer. A distribute transaction transfers money from a single account to $N$ other accounts; A collect transaction transfers money from $M$ accounts to a single account; A many-to-many transactions transfers money from $N$ accounts to another $M$ accounts. The benchmark is implemented on an OLTP-benchmark testbed for relational databases \cite{difallah2013oltp}.

We now explain the process of generating a workload of $m$ transactions. We use four parameters to characterize the workload: 1) the inter-transaction dependency probability threshold $\beta$, 2) the maximum number of dependent transactions $Tx_{max}$, 3) the amount of transferred money $\gamma$, and 4) the maximum transaction size $Size_{max}$. First, the inter-transaction dependency is modeled using PG. PG is constructed using the Erd\"{o}s-Renyi model \cite{erdos1960evolution}, in which an edge has probability $p$ of existence. The degree of inter-transaction dependency is controlled by grouping transactions into small groups. Within each group, a pair of transactions are dependent if $p$ is greater than $\beta$. Semantically, each group models money transfer transactions within a city, state, or country. In order to avoid fully connected graphs within the group of transactions, the number of allowed dependent transactions is limited to $Tx_{max}$. In other words, a transaction can be dependent to no more than $Tx_{max}$ transactions. Once a PG is created, tuples are assigned to the transactions as follows. The size of each transactions is determined using a uniform distribution with a range of [2,$Size_{\text{max}}$]. The transaction type is randomly selected from distribute, collect, or many-to-many. Each transaction shares a single tuple with the set of its dependent transactions. The value of $\gamma$ is chosen from a uniform distribution with range of [0.01,0.1]. Accordingly, the query is generated. Each transaction contains a single query. We generate 4 workloads with different sizes, mainly, 5000, 10000, 15000, and 20000 transactions. For each workload, we vary the value of $\beta$ to be 0.25, 0.5, and 0.75.

\subsection{Evaluation of the IB Demarcation Heuristics}
\label{sec:heurPer}

In this section, we evaluate the performance of the proposed heuristics using a synthetic transactional workload. We compare the performance of the proposed heuristics with two assignment techniques: Random Assignment (RA) and Skewed Assignment (SA). In RA, transactions are assigned to the IBs randomly. In SA, transaction are randomly assigned to the IBs following an 80-20 rule, i.e., 80\% of the transactions are assigned to 20\% of the IBs. SA is used to emphasize the importance of the size balancing objective (Equation \ref{eq:f2_1}). We use two metrics to evaluate the performance of the heuristics: the total number of boundary tuples (Equation \ref{eq:f1_2}), and the assignment fairness index. The fairness index is used to assess the performance in terms of size balancing among IBs. The fairness index is an adaptation of the Jain's fairness index \cite{jain1999throughput} in the following way. 
\begin{equation}
\mathcal{J}(IB_1,IB_2,\ldots,IB_k) = \frac{(\sum_{i=1}^k |IB_i|)^2}{k \cdot \sum_{i=1}^k |IB_i|^2}
\end{equation}
where $|IB_i|$ is the number of transactions assigned to the i$th$ IB.

In the first experiment, we vary the number of IBs, k, using a workload of 5000 transactions, and $\beta=0.75$. The number of boundary tuples increases as $k$ increases (refer to Fig. \ref{fig:ddg_ibs_bo}). The reason is that as $k$ increases, the number of contained transactions in each IB is reduced. As a result, the overlap among IBs increases and thus increases the total number of boundary tuples. We observe that BFA produces a less number of boundary tuples compared to BA and RA. The reason is that BFA performs a greedy decision to reduce the number of boundary tuples when assigning a transaction to IBs. Moreover, BA performs better than RA because BA attempts to balance the sizes of the IBs. Notice that SA produces less boundary tuples than BFA for $k$ $<$10, but more boundary tuples when k$>$10. The reason is that, when the number of IBs is small, e.g., for k=5 and k=10, 80\% of the transactions are assigned only to 1 IB and 2 IBs, respectively. Thus, most of the inter-transaction dependencies are contained in at most 2 IBs. BFA and BA produce fair assignment as in Fig. \ref{fig:ddg_ibs_fi}. SA performs the worst by design, while RA performs better when the number of transactions is large and the number of IBs is small because the random assignment tends to be fair. 

In the next experiment, we vary the number of transactions and the value of $\beta$ while generating the IB assignment with $k= 10$. We observe that increasing the number of transactions increases the number of boundary tuples as in Fig. \ref{fig:ddg_tx_bo}. The reason is that increasing the number of transactions increases the number of shared tuples among transactions. Notice that the set of boundary tuples is a subset of the shared tuples. Similarly, increasing the value of $\beta$ increases the number of boundary tuples as in Fig. \ref{fig:ddg_alpha_bo}. The reason is that as the value of $\beta$ increases, the number of dependent transactions increases. As a result, the number of shared tuples increases. 

\begin{figure}[t!]
  \centering
  \subfigure[Number of boundary tuples]{\label{fig:ddg_ibs_bo}  \includegraphics[width=0.5\columnwidth]{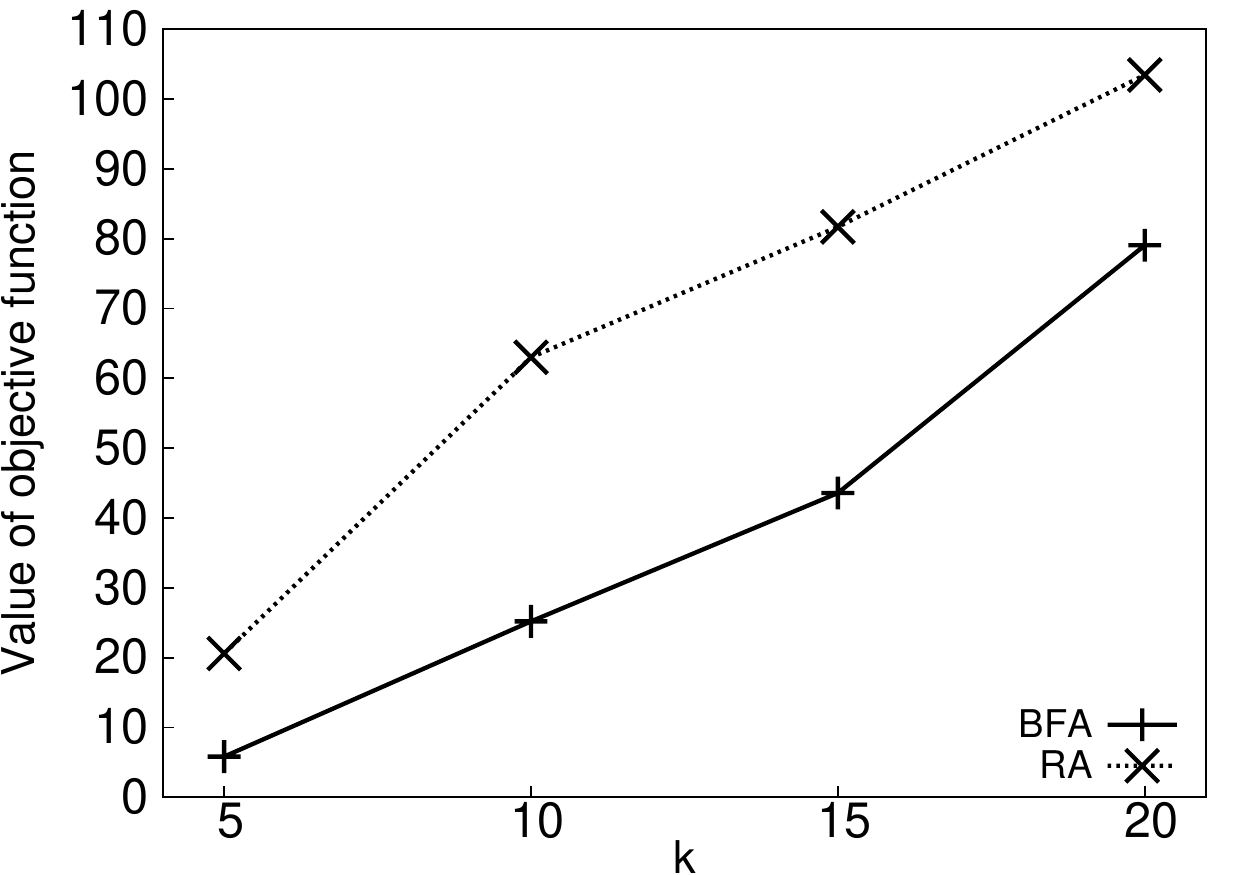}}%
    \subfigure[Fairness index]{\label{fig:ddg_ibs_fi}\includegraphics[width=0.5\columnwidth]{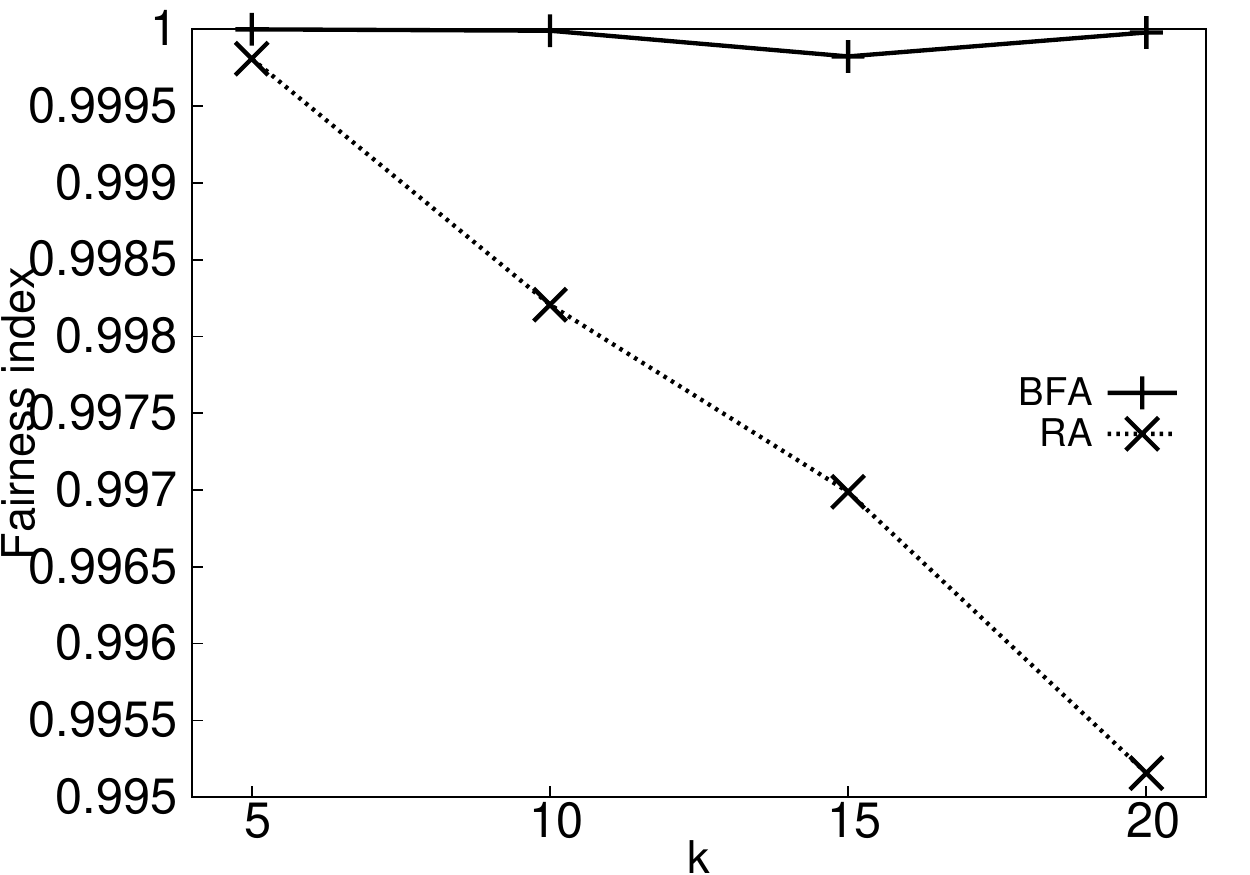}}
  \caption{Effect of $k$ on a workload of 5000 transaction with $\alpha=0.75$.}
  \label{fig:ddg_ibs}
                      \vspace{-3 mm}
\end{figure}

\begin{figure}[t!]
  \centering
  \subfigure[Effect of the number transactions]{\label{fig:ddg_tx_bo}  \includegraphics[width=0.5\columnwidth]{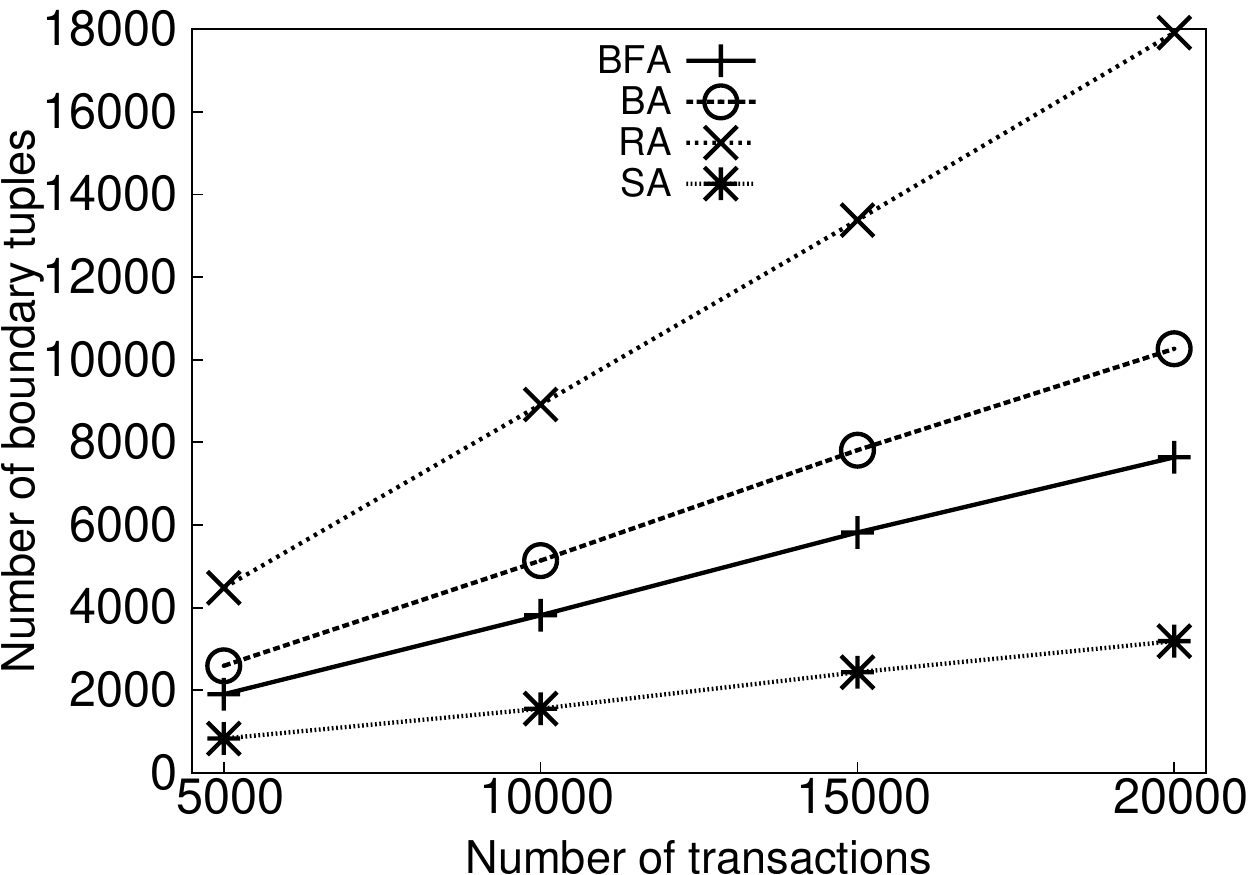}}%
%  \subfigure[Fairness index]{\label{fig:ddg_tx_fi}\includegraphics[width=0.5\columnwidth]{././figures/gnuplot/heuristics/size/fair.eps}}
  \subfigure[Effect of $\beta$]{\label{fig:ddg_alpha_bo}  \includegraphics[width=0.5\columnwidth]{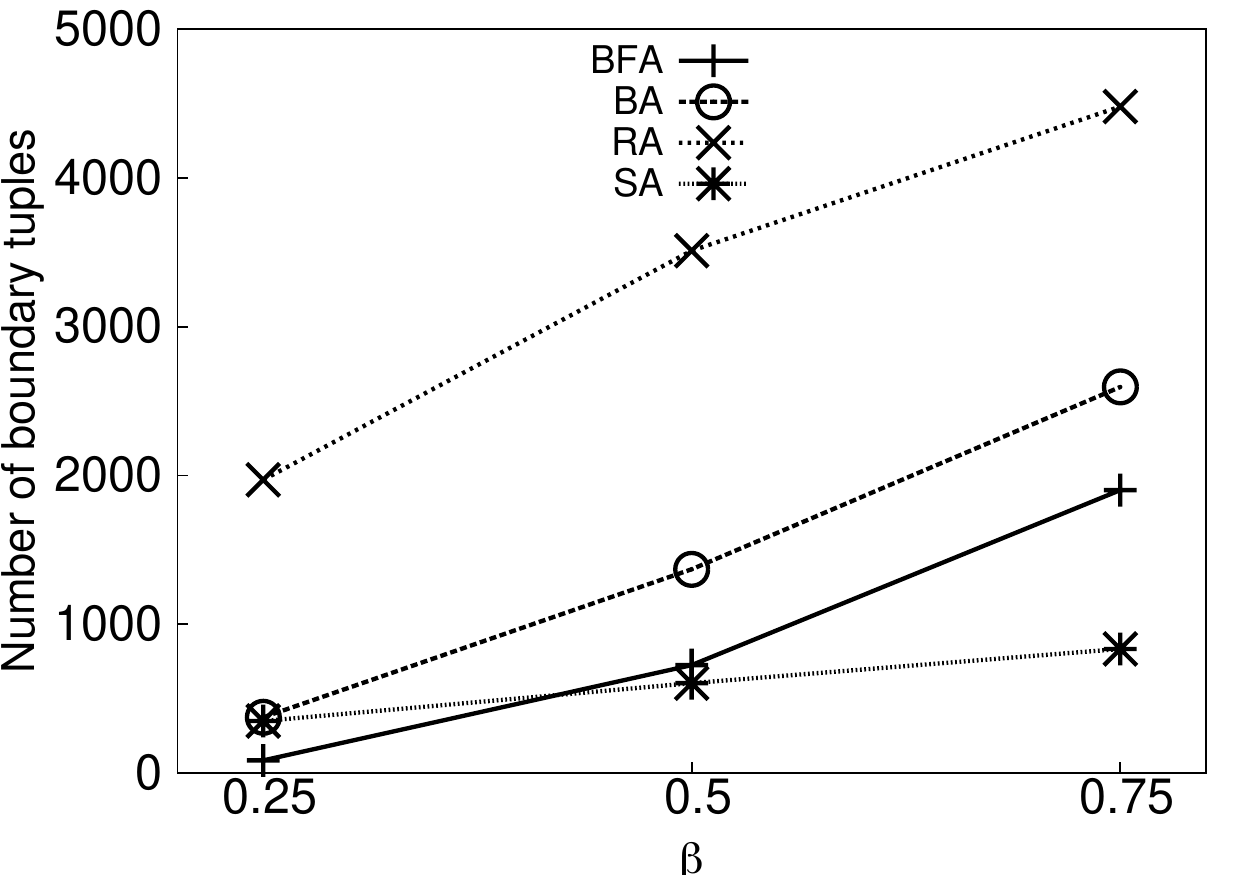}}%
  \caption{Effect of the number of transactions and value of $\beta$ on a workload 5000 transaction with $\alpha=0.75$.}
  \label{fig:ddg_tx}
                      \vspace{-3 mm}
\end{figure}

%\begin{figure}[t!]
%  \centering
%  \subfigure[Number of boundary tuples]{\label{fig:ddg_alpha_bo}  \includegraphics[width=0.5\columnwidth]{././figures/gnuplot/heuristics/connectivity/bo.eps}}%
%  \subfigure[Fairness index]{\label{fig:ddg_alpha_fi}\includegraphics[width=0.5\columnwidth]{./figures/gnuplot/heuristics/connectivity/fair.eps}}
%  \caption{Varying the values of $\alpha$ for 5000 transactions}
%  \label{fig:ddg_alpha}
%\end{figure}

%\begin{figure}[t!]
%  \centering
%  \subfigure[Number of boundary tuples]{\label{fig:at_o_ddg2}  \includegraphics[width=0.5\columnwidth]{./figures/gnuplot/heuristics/locality/bo.eps}}%
%  \subfigure[Fairness index]{\label{fig:at_h_ddg2}\includegraphics[width=0.5\columnwidth]{./figures/gnuplot/heuristics/locality/fair.eps}}
%  \caption{}
%  \label{fig:at_ddg2}
%\end{figure}

\subsection{PIMS Performance Evaluation}

The experiments are conducted on a Dell Power Edge R420 server with 6-core Intel E5-2620v3 CPU, 120 GB of RAM, and Ubuntu Server 16.0 LTS OS. We use PostgreSQL 9.5. PIMS is implemented using Java and is connected to PostgreSQL using JDBC. We use 8 workers (threads) for the admission controller and an unlimited pool of threads for the response and recovery procedures. Before each experiment, the Checking table is populated with 100,000 tuples of unique IDs and initial balance of \$10,000. We use a B-tree index on the Checking table to improve the query execution time. We assume that the transactions are submitted according to a Poisson distribution with an arrival rate of $\lambda$. Furthermore, malicious transactions are injected into the workload using a uniform distribution over time. The number of malicious transactions is based on the attack intensity $\pi$ that is a percentage of the total number of transactions.

We evaluate the performance of PIMS using four metrics: 1) the number of affected transactions, 2) the number of blocked transactions, 3) the average recovery time, and 4) the average response time. The number of blocked transactions indicates the performance of PIMS in terms of availability, while the number of affected transactions indicates the cost of damage. The average response time is computed based on the response times of all committed transactions. We compare the performance of PIMS using the proposed heuristics, i.e., BFA and BA, with RA and SA. We refer to PIMS with k$>$1 as "PIMS\_k". We often use PIMS\_k\_BFA, PIMS\_k\_BA, PIMS\_k\_RA, and PIMS\_k\_SA to differentiate between the assignment techniques as needed. The performance of PIMS is compared against PIMS\_1 (denoted by \textit{One IB}) and ITDB \cite{bai2009data}. We use a workload of 5000 transactions to conduct the experiments, unless stated otherwise. In each experiment, we run PIMS with $k$ values of 5, 10, 15, and 20. We do not notice any improvement in the performance beyond 20 IBs. 

\subsubsection{Read Log Overhead}

In this experiment, we study the overhead of logging the transaction read/write operations. Figs. \ref{fig:readoverhead_th} and \ref{fig:readoverhead_res} give the throughput and response time with logging (PIMS) and without logging (no logging) for various values of $\lambda$. Notice that the throughput in both cases match the transaction arrival rate. However, logging read/write operations adds 30\% overhead to the response time. This overhead is inevitable because PIMS relies on the transactions log to identify AT, and generate the compensating transactions.

\begin{figure}[t!]
  \centering
  \subfigure[Throughput]{\label{fig:readoverhead_th}  \includegraphics[width=0.5\columnwidth]{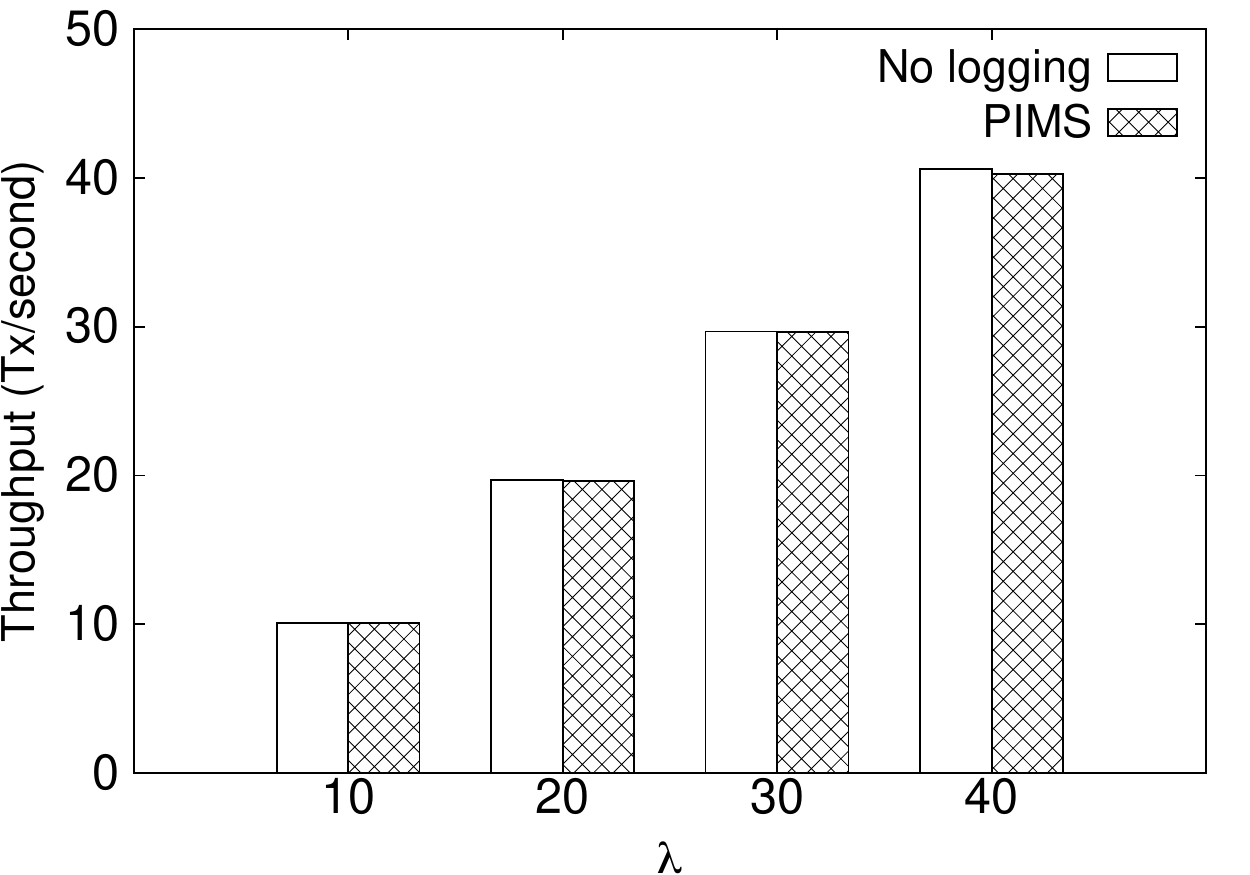}}%
  \subfigure[Response time]{\label{fig:readoverhead_res}\includegraphics[width=0.5\columnwidth]{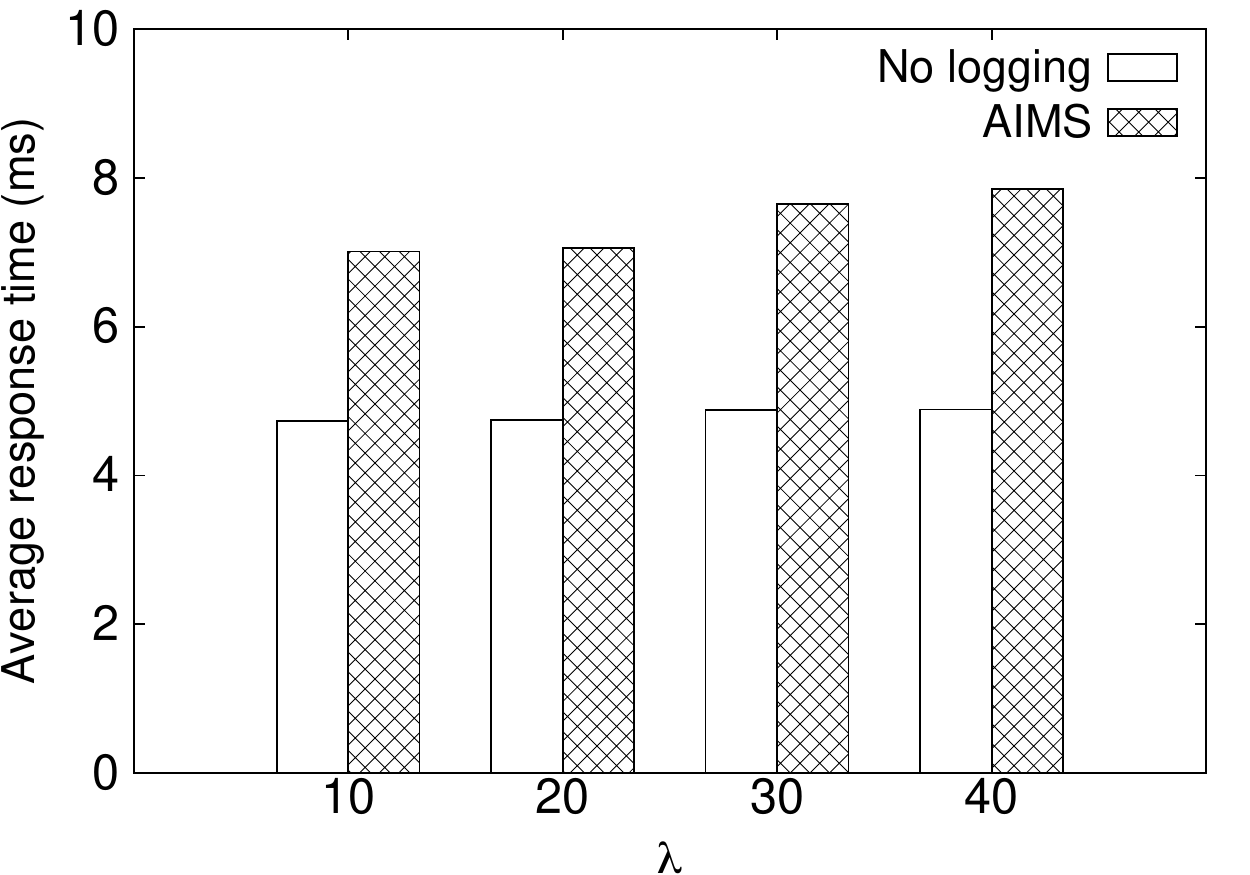}}
  \caption{Overhead of logging read operations.}
  \label{fig:readoverhead}
                      \vspace{-3 mm}
\end{figure}

\subsubsection{Effect of Attack Intensity}

In this experiment, we compare the performance of PIMS and OneIB with ITDB as the attack intensity $\pi$ increases. We use $k=10$ and $\Delta = 100 \ ms$ to plot the results using BFA and BA assignments. Refer to Fig. \ref{fig:numMal}. We observe that the extent of the damage increases as the attack intensity increases. In particular, the number of affected transactions increases as the number of malicious transactions increases. This is intuitive because increasing the number of malicious transactions increases the number of affected transactions with the same inter-transaction dependency. However, the average time to recover the damage remains constant as in Fig. \ref{fig:numMal_rec}. Observe that the number of affected transactions and the average recovery time for PIMS\_10\_BFA and PIMS\_10\_BA are less than OneIB and ITDB. The reason is that PIMS employs the delayed access mechanism that takes a proactive approach to block the incoming transactions that can potentially spread the damage. However, the disadvantage of the delayed access mechanism is the increase in the number of blocked transactions and average response time as in Fig. \ref{fig:numMal_blk} and \ref{fig:numMal_res}. PIMS\_BFA has less response time as compared to PIMS\_BA because BFA generates lower number of boundary tuples. Notice that the delayed access mechanism is not active in the case of OneIB and ITDB because there are no boundary tuples, and thus the response time is lower than PIMS. 

Notice that, in general, increasing the number of affected transactions increases the number of compensating transactions to be performed. Consequently, the average recovery time increases. Similarly, increasing the number of blocked transactions increases the average response time. Even though OneIB and ITDB almost have the same number of affected transactions with the same attack intensity, OneIB encounters less recovery time as compared to ITDB. The reason is that ITDB needs multiple passes on the transactions logs to find AT \cite{ammann2002recovery}, while OneIB needs a single pass because it temporarily blocks transactions. However, the overhead of blocking incoming transactions is an increase in the number of blocked transactions and response time in OneIB as compared to ITDB. 

We note that the proposed response and recovery methodology, i.e., PIMS, reduces the number of affected transactions by at least 33\% as compared to ITDB. Consequently, the average recovery time is reduced by 50\%, i.e., 150 $ms$ in PIMS\_10\_BFA as compared to 300 $ms$ in ITDB with 750 malicious transactions. The downside is that PIMS incurs a larger number of blocked transactions due to blocking the transactions that request to read boundary tuples. The increase in the response time is around 60\% as compared to ITDB. Nevertheless, the average response time for PIMS does not exceed 50 $ms$ when the attack intensity is 15\%. Fig. \ref{fig:numMal_th} gives the throughput of PIMS and OneIB as the attack intensity increases for $\Delta=100$ $ms$ and $\lambda=10$. We note that PIMS and OneIB match the transaction arrival rate $\lambda$, and thus does not incur any overhead on the throughput. 

\begin{figure}[t!]
  \centering
  \subfigure[Affected transactions]{\label{fig:numMal_aff}  \includegraphics[width=0.5\columnwidth]{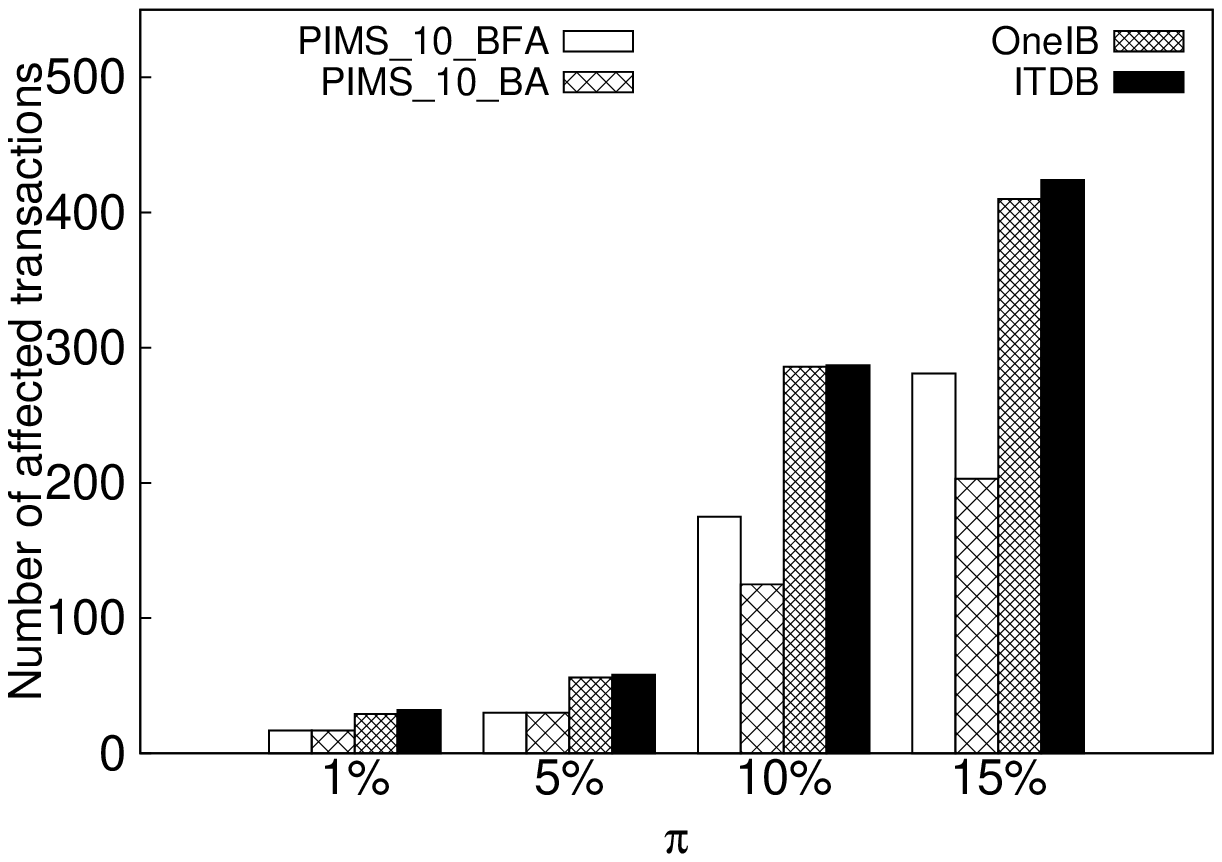}}%
  \subfigure[Blocked transactions]{\label{fig:numMal_blk}\includegraphics[width=0.5\columnwidth]{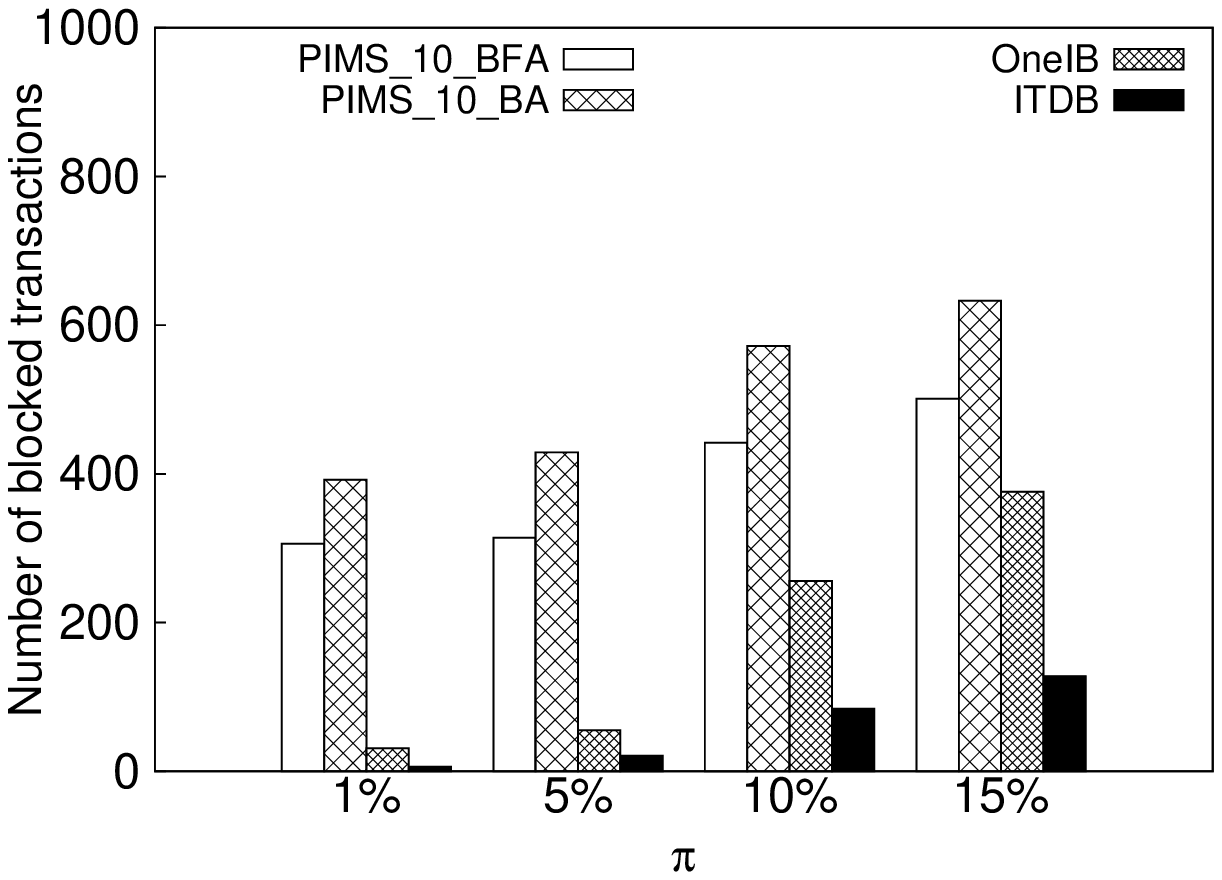}}
\subfigure[Average recovery time]{\label{fig:numMal_rec}\includegraphics[width=0.5\columnwidth]{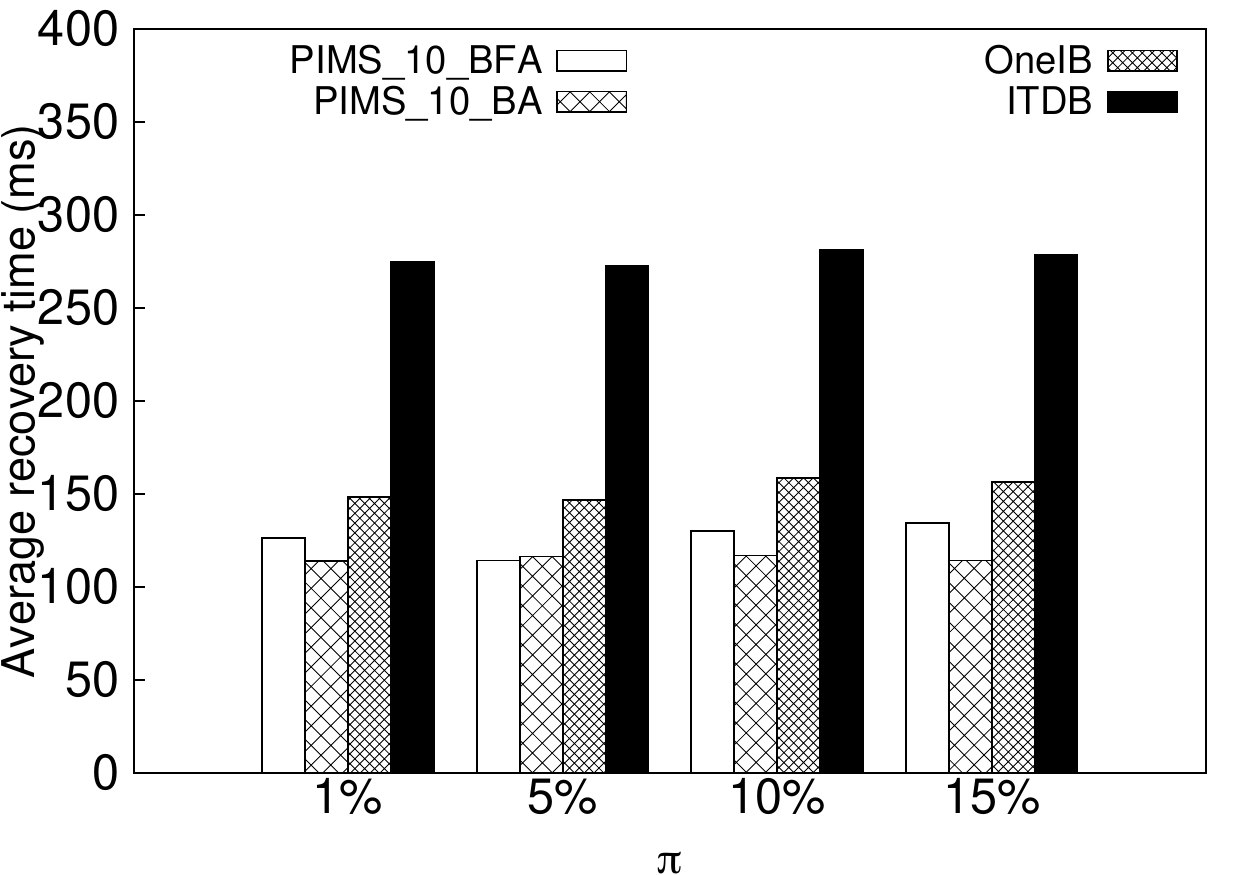}}%
  \subfigure[Average response time]{\label{fig:numMal_res}\includegraphics[width=0.5\columnwidth]{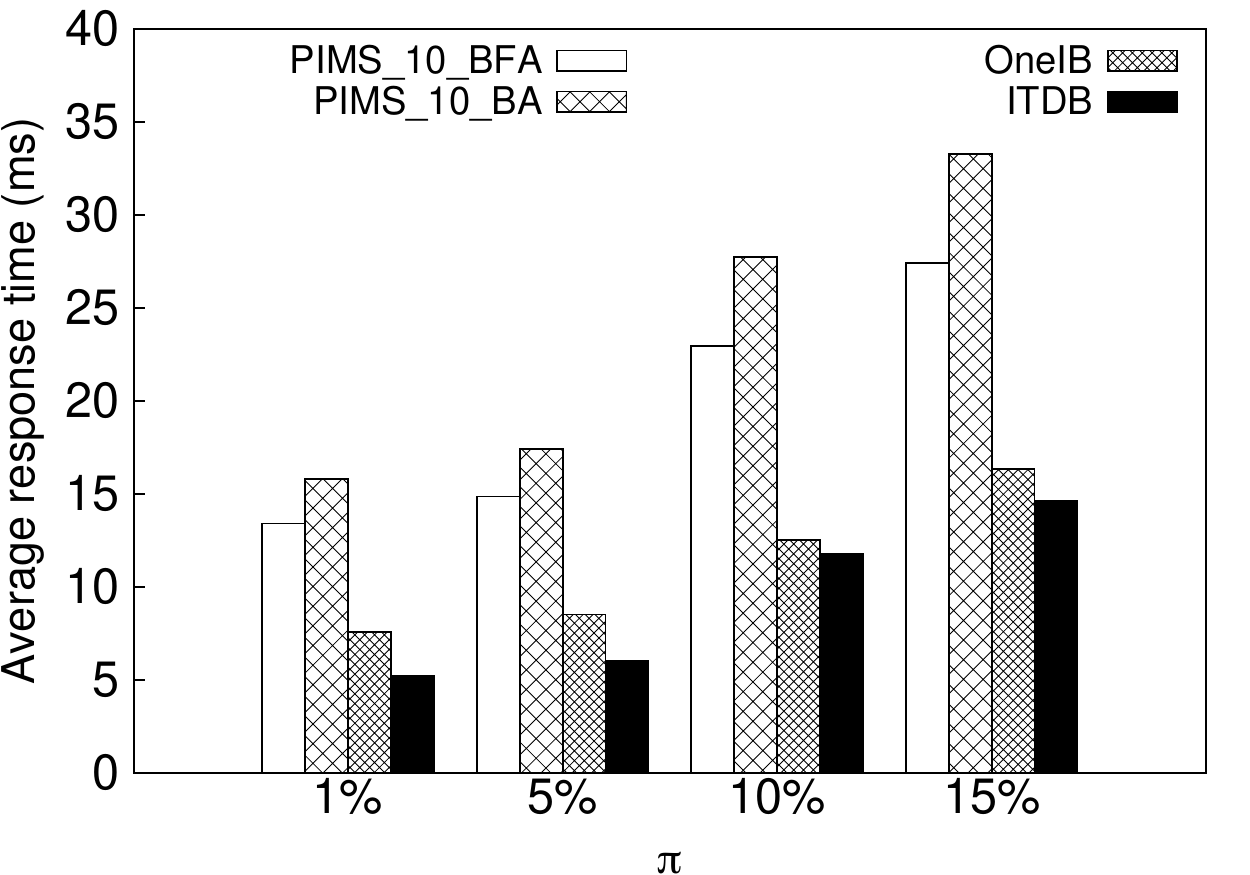}}%
  \caption{Effect of $\pi$ with $\Delta = 100$ $ms$ and $\lambda =10$.}
  \label{fig:numMal}
                      \vspace{-3 mm}
\end{figure}

\begin{figure}[t!]
        \centering
		\includegraphics[width=0.5\columnwidth]{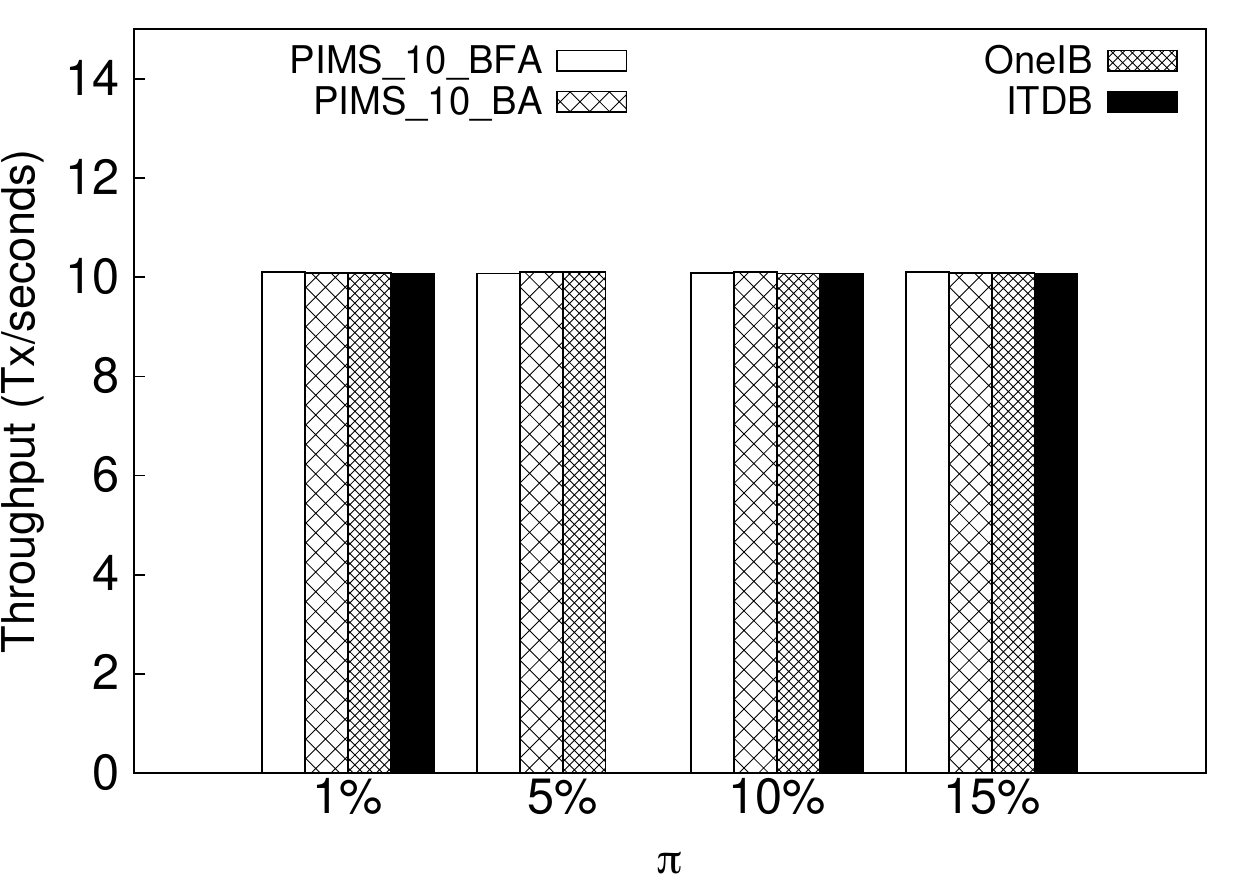}
                \caption{Throughput of PIMS and OneIB for different values of ($\pi$) with $\Delta = 100$ $ms$ and $\lambda =10$.}
                  \label{fig:numMal_th}
                                      \vspace{-3 mm}
\end{figure}

\subsubsection{The Effect of The Number of IBs}
\label{exp:ibs}

\begin{figure}[th!]
  \centering
  \subfigure[Affected transactions]{\label{fig:ibs_d100_l10_aff}  \includegraphics[width=0.5\columnwidth]{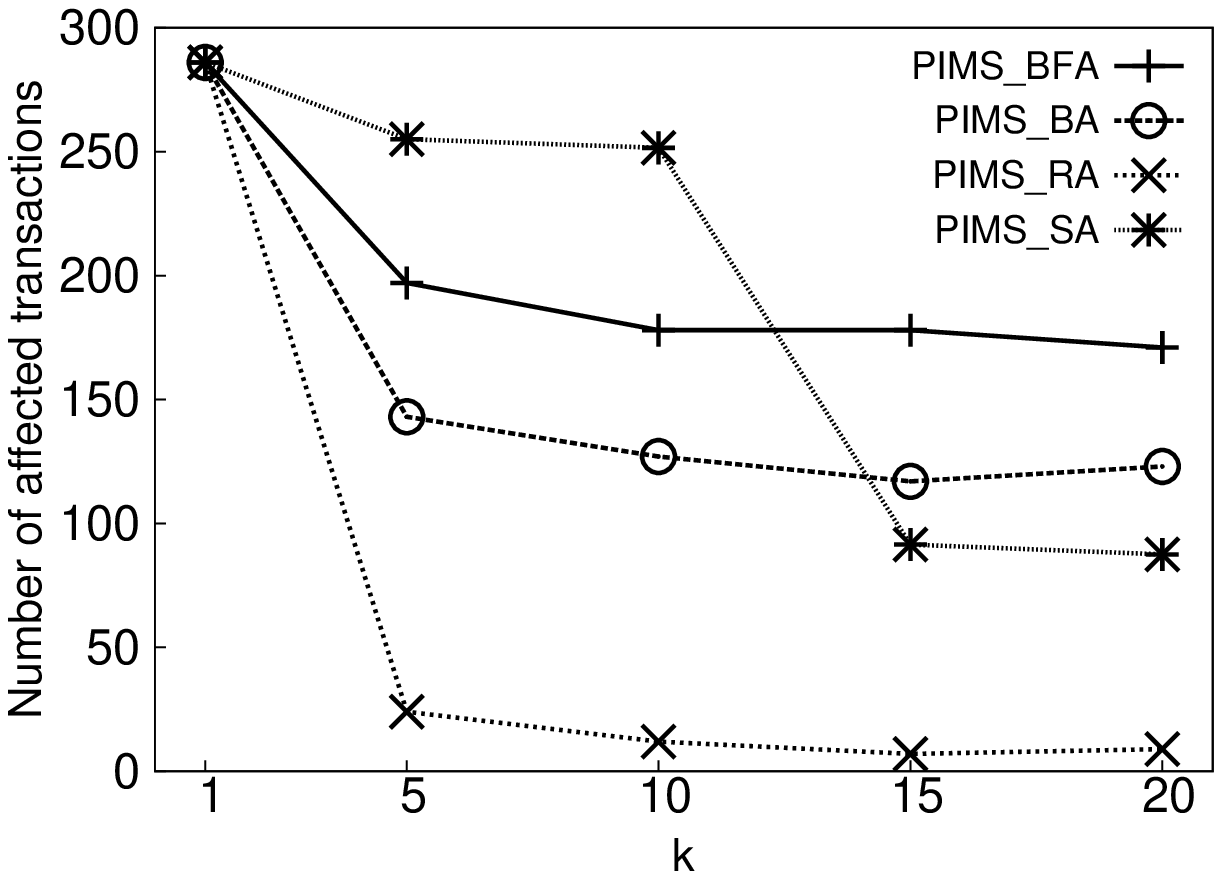}}%
  \subfigure[Blocked transactions]{\label{fig:ibs_d100_l10_blk}\includegraphics[width=0.5\columnwidth]{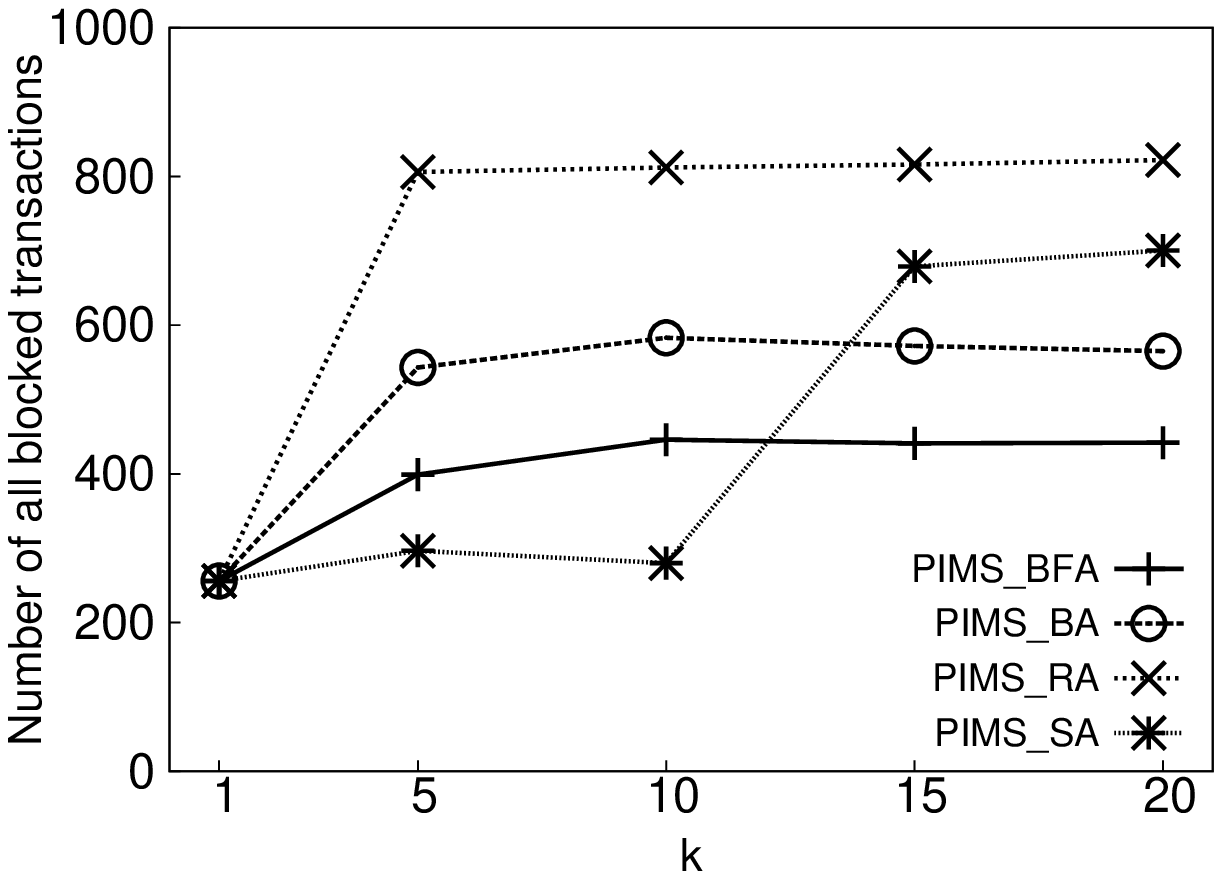}}
      \subfigure[Average recovery time]{\label{fig:ibs_d100_l10_rec}\includegraphics[width=0.5\columnwidth]{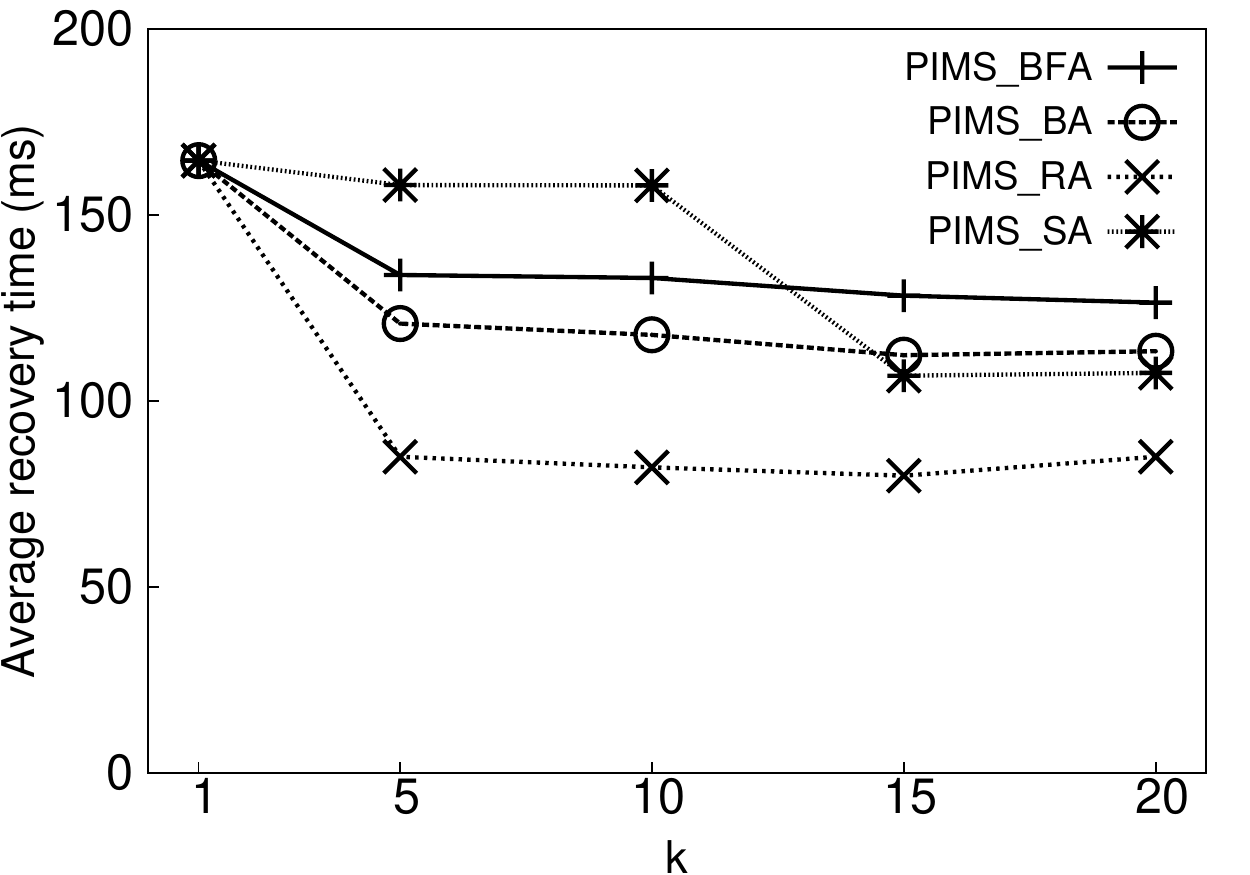}}%
  \subfigure[Average response time]{\label{fig:ibs_d100_l10_res}\includegraphics[width=0.5\columnwidth]{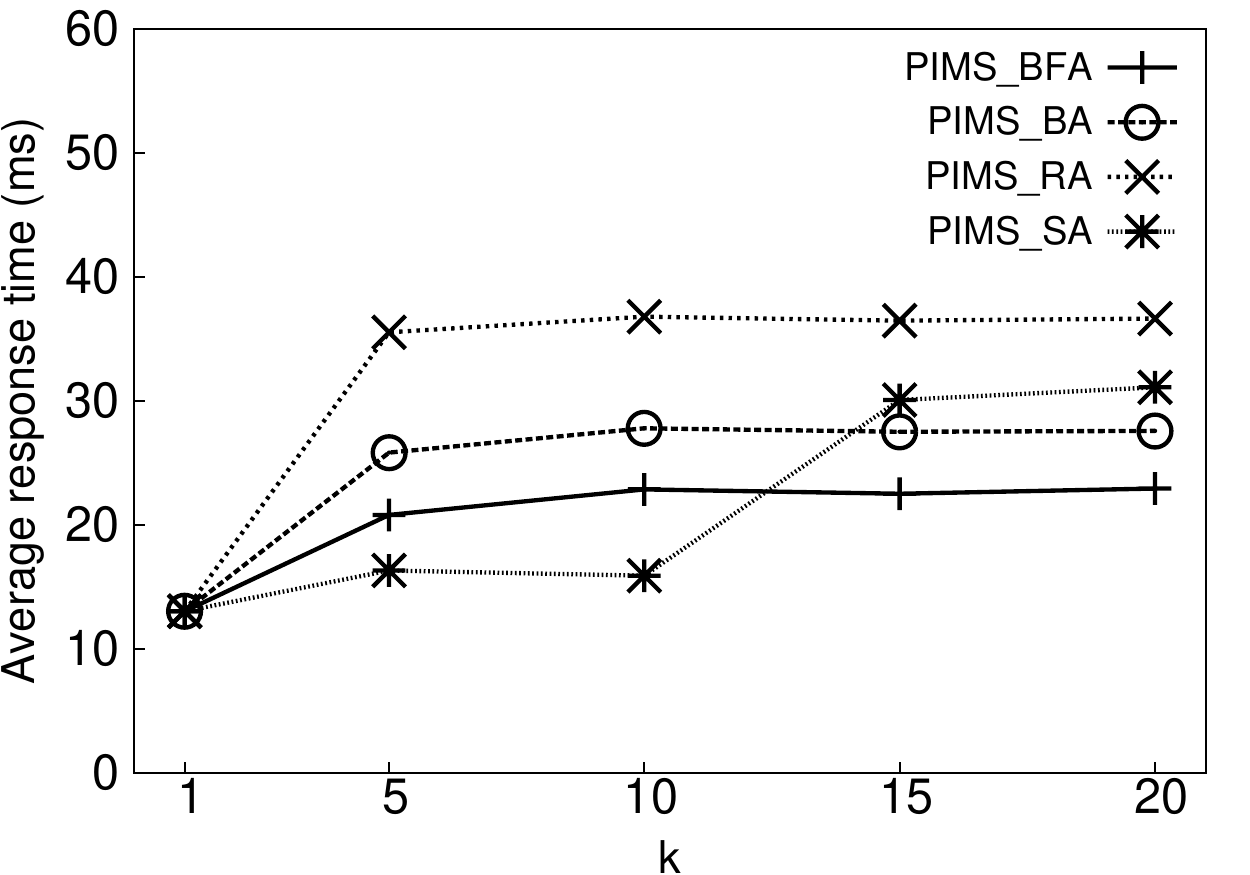}}%
  \caption{Effect of $k$ with $\lambda=10$, $\Delta =100$ $ms$, and $\pi=10\%$ malicious intensity.}
  \label{fig:ibs_d100_l10}
                      \vspace{-3 mm}
\end{figure}
In this experiment, we evaluate the performance of PIMS with various values of $k$, i.e., 1, 5, 10, 15, and 20. We compare the results of PIMS\_BFA and PIMS\_BA with oneIB. Moreover, we study the performance of PIMS\_RA and PIMS\_SA to demonstrate the effectiveness of the proposed heuristics. Fig. \ref{fig:ibs_d100_l10} gives the performance of PIMS with $ \Delta=100 $ $ms$, $\lambda = 10$, and $\pi=10\%$, i.e., 500 malicious transactions. From Fig. \ref{fig:ibs_d100_l10}, increasing $k$ improves the performance of PIMS. In particular, the number of affected transactions is reduced by at least 5\%, 18\%, 26\%, and 48\% using SA, BFA, BA, and RA, respectively. Consequently, the average recovery time decreases as $k$ increases. However, PIMS incurs at least 20\%, 37\%, 50\%, and 63\% of response time overhead in SA, BFA, BA, and RA, respectively, as Fig. \ref{fig:ibs_d100_l10_res} illustrates. The reason is due to the increase in the number of boundary tuples. We note that the increase in the response time overhead is marginal as the value of $k$ increases beyond 10. 

Although PIMS\_BFA has the least reduction in recovery time as compared to PIMS\_BA and PIMS\_RA, PIMS\_BFA has the minimum response time overhead. We note that PIMS\_SA has less response time than PIMS\_BFA when $k<10$ because the number of boundary tuples generated by SA is less than BFA as explained in Section \ref{sec:heurPer}. However, the response time of PIMS\_SA increases dramatically for k$>$10 as the assignment skewness is higher. On the other hand, PIMS\_RA outperforms PIMS\_BFA and PIMS\_BA in terms of reducing the number of affected transactions, and thus the recovery time. PIMS\_RA incurs the highest overhead in response time. In conclusion, we note that PIMS\_BFA and PIMS\_BA produce a balanced performance in terms of recovery and response time as compared to PIMS\_RA and PIMS\_SA over a wide range of $k$ values.

In the next experiment, we study the effect of the delayed access mechanism on the performance of PIMS. We compare the performance of PIMS\_BFA, PIMS\_BA, PIMS\_RA, and PIMS\_SA with delay and without delay. From Fig. \ref{fig:noDelay_aff}, the number of affected transactions is larger when the delayed access mechanism is off. The reason is that the damage propagates across the IBs. However, the shortcoming of the delayed access mechanism is the increase in the response time overhead as in Fig. \ref{fig:noDelay_res}. The reason is the increase in the number of blocked transactions due to blocking boundary tuples. The overhead is more noticeable in the case of PIMS\_RA as compared to PIMS\_BFA. In conclusion, the delayed access mechanism allows to contain the damage and reduce the recovery time in the case of PIMS\_BFA and PIMS\_BA. Moreover, PIMS with delayed access mechanism maintains a reasonable response time overhead, and thus improves the overall availability. 

\begin{figure}[t!]
  \centering
  \subfigure[Affected transactions]{\label{fig:noDelay_aff}  \includegraphics[width=0.5\columnwidth]{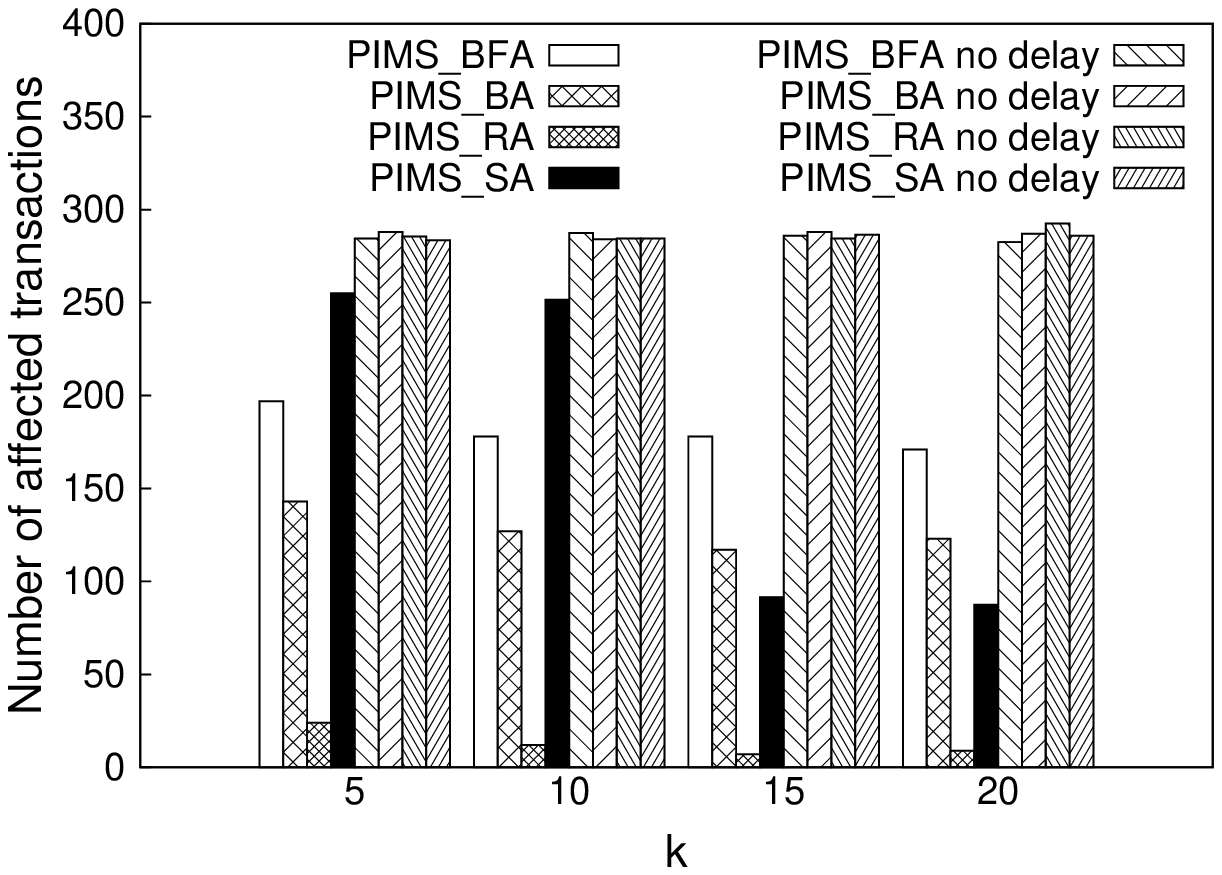}}%
  \subfigure[Blocked transactions]{\label{fig:noDelay_blk}\includegraphics[width=0.5\columnwidth]{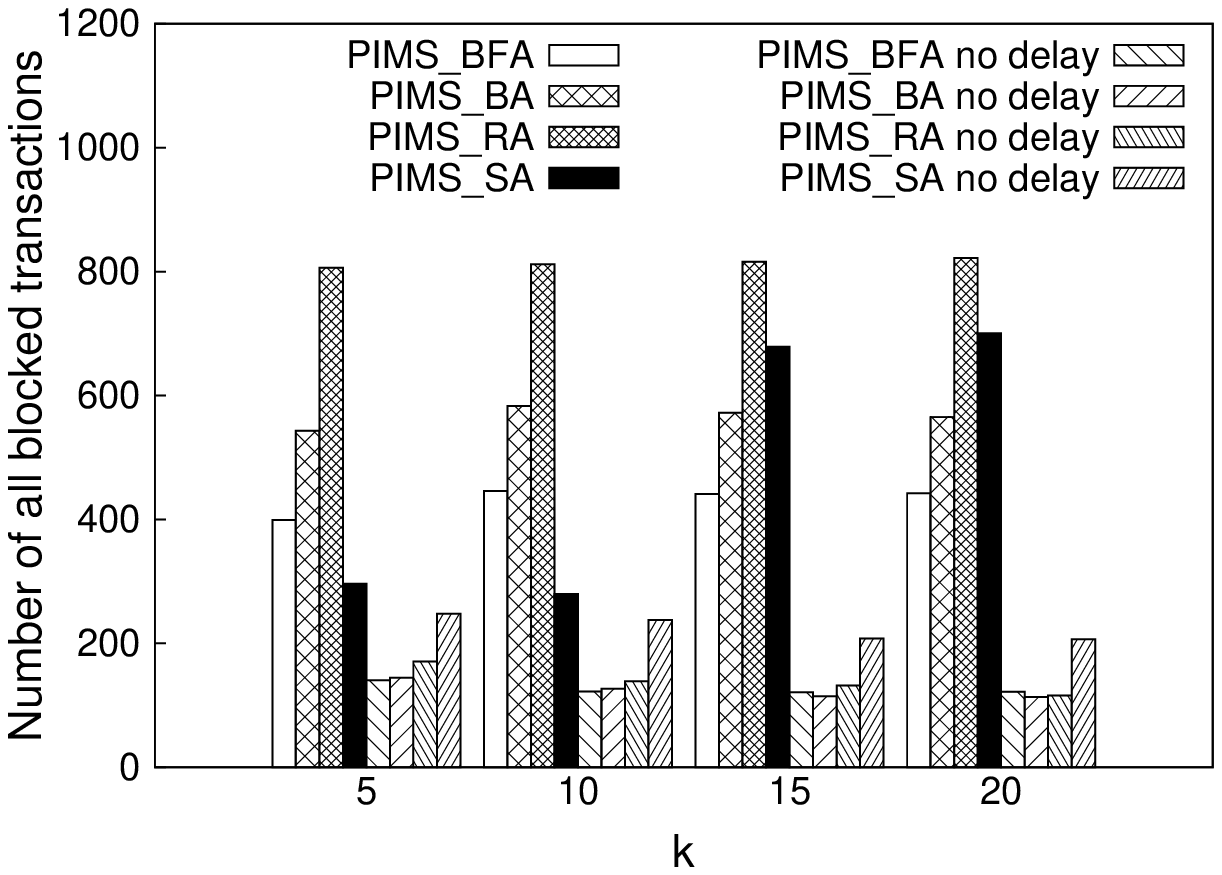}}
\subfigure[Average recovery time]{\label{fig:noDelay_rec}\includegraphics[width=0.5\columnwidth]{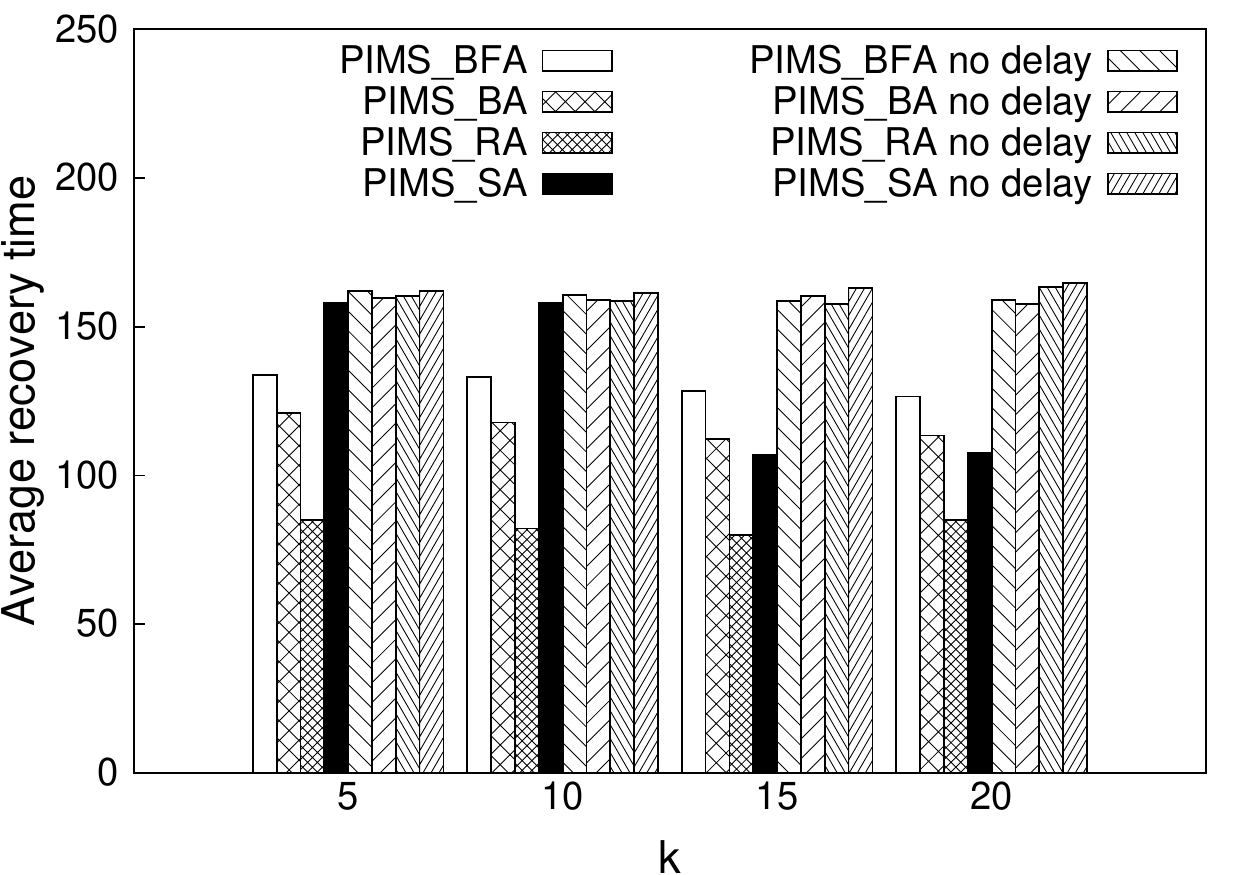}}%
  \subfigure[Average response time]{\label{fig:noDelay_res}\includegraphics[width=0.5\columnwidth]{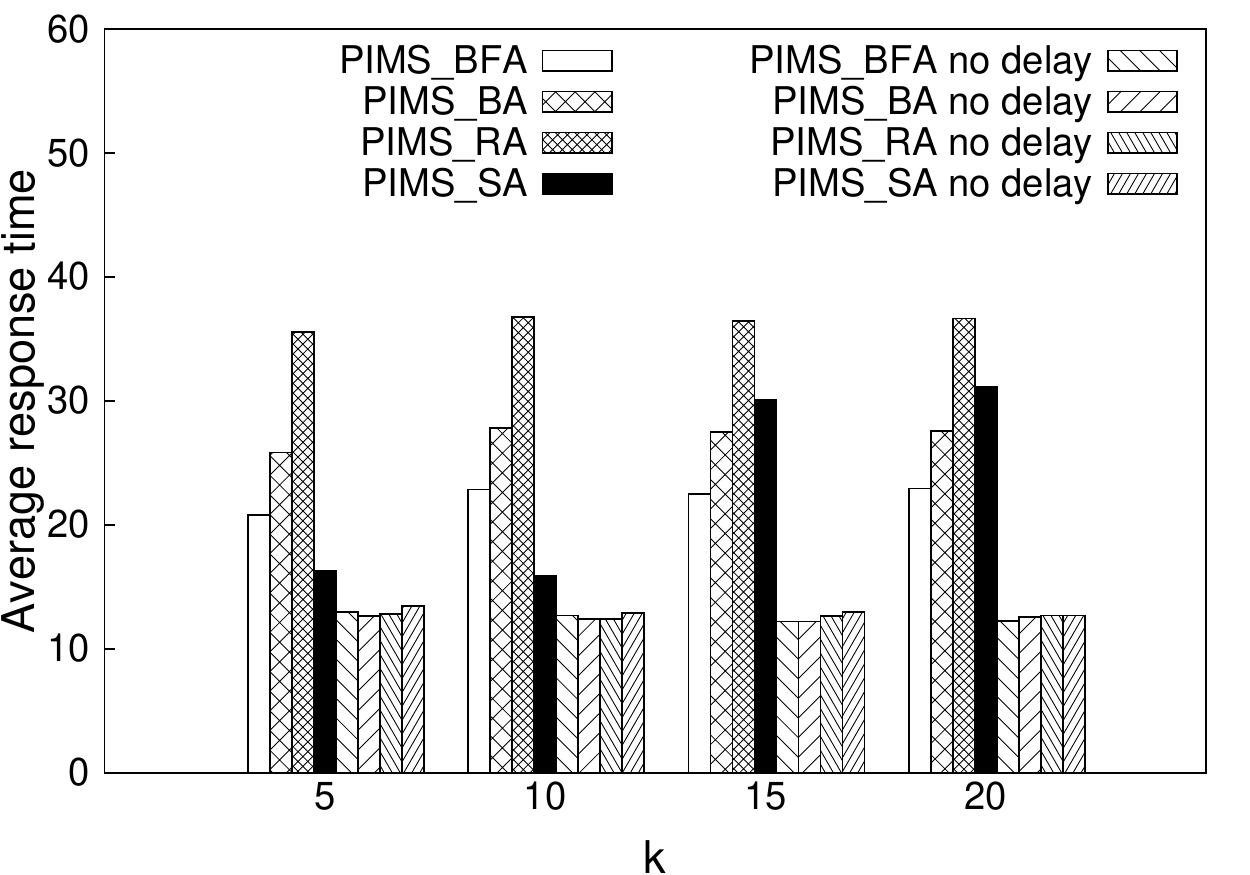}}%
  \caption{Effect of the delayed access mechanism with $\Delta=100$ $ms$, $\lambda=10$, and $\pi=10\%$.}
  \label{fig:noDelay}
                      \vspace{-3 mm}
\end{figure}

\section{Related Work}
\label{sec:related}

Security measures in DBMSs includes the protection of data confidentiality, integrity, and availability \cite{nist2002risk,bertino2005database}. A broad span of research addresses the protection of data confidentiality in DBMSs including authorization, e.g., \cite{chaudhuri2007fine}, access control, e.g., \cite{sandhu1996role}, encryption, e.g., \cite{sarfraz2015dbmask}, and inference and disclosure control, e.g., \cite{chen2008protection}. The integrity risk data in DBMSs is jointly prevented by using access control mechanisms and semantic integrity constraints to verify the correctness of the database state after updates\cite{bertino2005database}. The availability of data is protected by providing fault-tolerance \cite{fault-tolerancegashi}, replication \cite{patino2005middle}, and intrusion detection techniques \cite{milenkoski2015evaluating}.

In the case of successful intrusion attacks, the effects of the committed transactions are undesirable. The methodology of undoing a committed transaction can be generally handled by one of two approaches: rollback or compensation. The rollback approach is achieved by rolling back all desirable and undesirable activities to a point that is free from damage \cite{mohan1992efficient}. On the other hand, compensation approach unwinds the effect of selective committed transactions by executing special compensating transactions. The compensation operations are either action-oriented or effect-oriented \cite{korth1990formal}. In this paper, we follow an effect-oriented compensation approach to recover from the damage caused by the malicious and affected transactions.

Several solutions have been proposed for intrusion recovery in database-backed applications. A generic intrusion-tolerant architecture for web-servers uses redundancy and diversification principles proposed in \cite{saidane2009design}. In \cite{chandra2011intrusion}, WARP is proposed to recover from intrusions in web-applications by rolling back the database and replaying subsequent legitimate actions to correct the state of the DBMS. In \cite{pardal2017rectify}, an intrusion recovery tool for database-backed applications running in Platform-as-a-Service clouds is proposed. The tool uses machine learning techniques to associate the application requests to the DBMS statements, and then uses existing recovery algorithms to recover from the damage in the DBMS. PIMS is designed as a middle-layer between the DBMS and the application that performs automatic intrusion response and recovery in the DBMS independently from the running applications.  

Previous work in intrusion recovery in DBMSs can be broadly classified into two categories: transaction-level and data-level approaches. In the transaction-level approach, the general direction is to selectively rollback or compensate for the damaged tuples. In \cite{ammann2002recovery}, a suite of recovery algorithms is proposed to unwind the effect of malicious transactions for offline and online recovery. In \cite{liu2004design,bai2009data}, the authors present ITDB and DTQR, respectively, that implement the recovery algorithms in \cite{ammann2002recovery} on top of a Commercial-Off-The-Shelf DBMS. In \cite{chiueh2008accurate}, a damage assessment and repair system, termed Phoenix, is introduced. The core component in Phoenix is the inter-transaction dependency tracking that maintains such persistent dependency information at run-time. On the other hand, data-dependency approach provides a flexible recovery at the object-level. In \cite{panda2002extended}, a damage assessment technique using data dependency analysis is proposed to obtain precise information about the set of corrupted data. PIMS uses a hybrid approach between data-level dependency and transaction-level approach to track the damage. In particular, the damage assessment is performed at the data-level, while the response and recovery procedures are performed at the transaction-level. Moreover, PIMS addresses the problem of prolonged online recovery procedure in \cite{ammann2002recovery}.

Date partitioning schemes are used to improve the availability and scalability of the DBMS. In \cite{curino2010schism}, a workload-aware approach for partitioning the data is proposed. The partitioning approach models the data objects as a graph that is then partitioned into $k$ balanced partitions such that the number of distributed transactions is minimized. In \cite{quamar2013sword}, a scalable workload-aware data placement that uses hyper-graph compression techniques to deal with large-scale datasets is proposed. Online partitioning techniques adaptively partition data based on emerging hotspots, workload skews, and load spikes. In \cite{turcu2016automated}, a methodology for using automatic data partitioning that prefers partitions with independent transactions is proposed. In \cite{taft2014store}, E-store is proposed that provides an elastic planning and reconfiguration system to mitigate the challenges paired with workload skews. None of the above partitioning scheme considers the security aspects of DBMS. The IB demarcation scheme partitions the workload with the objective to improve the availability by confining the damage caused by intrusion attacks.

\section{Conclusion}
\label{sec:conc}
In this paper, the problem of response and recovery of successful intrusion attacks on a DBMS is addressed. We propose PIMS, a data partitioning-based intrusion management system for DBMSs, that can endure intense malicious intrusion attacks. A new fine-grained dependency model that captures the intra-transaction and inter-transaction dependencies is introduced. We introduce a data partitioning scheme, termed IBs, with the objective to limit the extent of the damage into partitions. We formulate the IB demarcation problem as a cost-driven optimization problem and prove that IBDP is NP-hard, and propose efficient heuristic solutions. We present the architecture of PIMS and conduct various experiments to evaluate its performance. We show that although PIMS incurs response time overhead, the reduction in the number of affected transactions and the recovery time is up to 48\% and 52\%, respectively. 

\bibliographystyle{unsrt}


% or


\ifCLASSOPTIONcompsoc
  
%\else
%  % regular IEEE prefers the singular form
  \section*{Acknowledgment}
This work is supported by Northrop Grumman.
%\fi

\ifCLASSOPTIONcaptionsoff
  \newpage
\fi

\bibliographystyle{IEEEtran}
\bibliography{references}

\clearpage
\end{document}